\documentclass[journal]{IEEEtran}
\usepackage{cite}
\usepackage{amsmath,amssymb,amsfonts}
\usepackage{amsthm}
\usepackage{graphicx}
\usepackage{xcolor}
\usepackage{algorithm}
\usepackage{algpseudocode}
\usepackage{bbm}
\usepackage{cuted}
\usepackage{mathtools}
\usepackage[hidelinks]{hyperref}
\usepackage{mdframed}
\usepackage{booktabs}

\theoremstyle{definition}
\newtheorem{theorem}{Theorem}

\newtheorem{proposition}[theorem]{Proposition}

\begin{document}
	\bstctlcite{IEEEexample:BSTcontrol}
	
	\title{\fontsize{20 pt}{\baselineskip}\selectfont Multi-User MIMO with Rotatable Antennas and IRS: Joint Antenna Boresight and IRS Orientation Design} 
	
	\author{Guoying~Zhang,
		Qingqing Wu,
		Ziyuan Zheng,
		Qiaoyan Peng,
		Ailing Zheng,
		Yanze Zhu,
		Ying Gao,
		Wen Chen
		\thanks{Guoying Zhang, Qingqing Wu, Ziyuan Zheng, Ailing Zheng, Yanze Zhu, Ying Gao and Wen Chen are with the Department of Electronic Engineering, Shanghai Jiao Tong University, Shanghai 200240, China (e-mail: \{gy\_zhang; qingqingwu; zhengziyuan2024; ailing.zheng; yanzezhu; Ying-gao; wenchen\}@sjtu.edu.cn). Qiaoyan Peng is with the State Key Laboratory of Internet of Things for Smart City, University of Macau, Macao 999078, China (e-mail: qiaoyan.peng@connect.um.edu.mo).}
	}

\maketitle

\begin{abstract}
	
	In this paper, we investigate an intelligent reflecting surface (IRS)-assisted multi-user system, where the base station (BS) employs rotatable antennas (RAs) and the IRS can adjust the panel orientation.
	To alleviate the severe multiplicative path loss of the cascaded channel, the IRS is deployed near the BS, while the user-BS and user-IRS links remain in the far field.
	We formulate a sum-rate maximization problem by jointly optimizing the receive beamforming, IRS phase shifts, BS antenna boresights, and IRS panel orientation.
	To tackle the resulting highly coupled and non-convex problem, we first study a  single-user case to reveal the structure of the dual-rotation gain, which is shown to be multiplicatively separable in the far field but coupled in the near field.
	For the general multi-user case, we develop an alternating optimization algorithm, where the receive beamforming is updated in closed form, the IRS phase shifts are optimized by an FP-assisted Riemannian conjugate gradient method, and the BS antenna boresights and IRS panel orientation are updated via projected gradient methods.
	Simulation results demonstrate the significant sum-rate gains achieved by the proposed coordinated rotation design over fixed-orientation and single-rotation benchmark schemes, and provide useful insights into near-field dual-rotation design.
	
\end{abstract}

\begin{IEEEkeywords}
	Intelligent reflecting surface (IRS), rotatable antenna (RA), near-field propagation, sum-rate maximization.
\end{IEEEkeywords}

\section{Introduction}
\label{sec:intro}

\IEEEPARstart{F}{uture} wireless networks are expected to support massive connectivity, high spectral efficiency, and reliable multi-user access under dense deployment scenarios~\cite{saad2020vision}.
Multi-user multiple-input multiple-output (MIMO) is a key technique for exploiting spatial-domain resources and improving system capacity~\cite{gesbert2007mimo_paradigm, larsson2014massive_mimo}.
In uplink transmission, users are usually located at different positions, and their signals arrive at the base station (BS) from distinct spatial directions, making directional array gain and inter-user spatial separability essential for reliable multi-user reception~\cite{tse2005fundamentals}.
However, conventional MIMO systems typically employ fixed antenna arrays, where the antenna positions and boresight directions are predetermined once deployed.
Although receive combining can adapt the signal processing weights in the digital or analog domain, it cannot directly reconfigure the element-level directional gain patterns.
This limits the ability of conventional MIMO systems to exploit the angular variations of uplink channels, thereby motivating spatially reconfigurable antenna architectures at the BS.

Among existing spatially reconfigurable architectures, fluid antenna systems (FASs) realize channel reconfiguration by switching among multiple ports or fluid states over a compact aperture~\cite{wong2021fluid,new2025fas_tutorial}. 
Different from FASs, movable antennas (MAs) explicitly adjust antenna positions to exploit the position-domain variation of wireless channels~\cite{zhu2025tutorial_comst,ma2024mimo}.
Specifically, MAs can move within prescribed local regions at the transmitter or receiver, so that their positions can be optimized to reshape the phase superposition of multipath components, thereby improving the received signal power, enhancing spatial multiplexing, and suppressing interference. 
Six-dimensional MAs (6DMAs) further extend the MA concept by jointly adjusting the three-dimensional positions and three-dimensional rotations of antennas or antenna surfaces~\cite{shao2025tutorial,shao20256d}. 
However, the enhanced flexibility of MAs, FASs, and 6DMAs is generally accompanied by increased hardware complexity, control overhead, and channel-acquisition burden.
By contrast, rotatable antennas (RAs) offer a simpler BS-side reconfiguration architecture, where the antenna phase centers remain fixed and only the boresight directions are adjusted~\cite{peng2025ra_spectrum}.

By adjusting antenna orientations, RAs can reconfigure the direction dependent antenna gains without changing the array geometry, thereby improving the effective array response with relatively low implementation overhead~\cite{ zheng2026rotatable_antenna_tutorial,zheng2025ra_modeling,zhang2026joint_ra_irs_lwc}. 
This feature is particularly useful for multi-user uplink transmission, where users may arrive at the BS from different AoAs and the receive gain of each RA can be adapted to match favorable signal directions. 
As a result, RAs can enhance useful signal components, improve the effective multi-user channel conditions, and mitigate interference from undesired directions.
Recent studies have optimized RA orientations for ISAC and UAV communications, demonstrating the effectiveness of orientation-domain adaptation for channel reconfiguration~\cite{zhou2025ra_isac,zhang2026ra_uav_mmwave,chen2025ra_uav_isac}.
Nevertheless, RAs can enhance the existing user-BS channel and thus provide limited performance improvement when the direct channel is blocked or severely attenuated.

To overcome this limitation, intelligent reflecting surfaces (IRSs) can be introduced to establish a cascaded user-IRS-BS channel~\cite{wu2024intelligent}.
An IRS consists of passive reflecting elements that impose controllable phase shifts on incident signals~\cite{wu2020Towards, wu2021tutorial}.
The combination of RAs and IRSs provides complementary control over the direct and cascaded channels, since BS boresight steering changes the BS-side directional response, while IRS phase shifts control the passive reflection over the cascaded user-IRS-BS channel.
Joint active and passive beamforming has been widely studied
to improve wireless channel quality~\cite{wu2019irs_joint}.
In addition, IRS-assisted systems with MAs have been studied~\cite{gao2025integrating,wu2025integrating}.
Movable intelligent surfaces (MISs) have also been proposed to reconfigure wireless channels by adjusting the relative positions of metasurface layers~\cite{zheng2025movable}.
However, existing MA-IRS and MIS designs mainly exploit position-domain reconfiguration and phase-shift optimization, while IRS panel orientation is not jointly considered.
The IRS panel orientation determines the direction-dependent reflection gain with respect to the incident and reflected directions. 
When the surface normal is poorly aligned with these directions, the cascaded channel suffers from additional reflection loss\cite{ozdogan2020physics}. 

To mitigate this direction-dependent loss, rotatable IRSs are proposed to adjust the panel orientation, thereby better aligning the surface normal with the incident and reflected directions~\cite{cheng2022ris_rotation,peng2025rotatable_irs}.
Recent studies have optimized IRS orientation for reflected-link gain enhancement and wireless coverage improvement, including joint IRS placement-and-orientation design for coverage enhancement~\cite{zeng2021ris_coverage}.
Cooperative rotatable IRSs have also been studied for joint beamforming and orientation optimization~\cite{peng2025cooperative_irs}. 
Orientation-aware IRS designs have also been investigated for mobile and UAV-assisted deployments, where the panel orientation needs to adapt to time-varying incident and reflected directions~\cite{jiang2026uav_ris}. 
In addition, the joint optimization of IRS rotation and 6DMA configuration under statistical channel state information was studied in~\cite{zhou2025rotatable}.
These studies demonstrate the benefits of IRS orientation control in improving the reflected link. 
For IRS-assisted multi-user MIMO uplink communications, the joint coordination of IRS orientation, IRS phase shifts, and BS antenna boresight steering remains insufficiently explored, despite its potential for more flexible cascaded-channel reconfiguration.

To reduce the severe path loss of the cascaded user-IRS-BS channel, IRSs are often deployed near the BS. 
Under such deployment, the user-BS and user-IRS links usually remain in the far field and can be approximated by plane-wave propagation~\cite{bjornson2020nearfield}, whereas
the short IRS-BS separation may place the IRS-BS link in the near field and spherical-wave propagation should be considered~\cite{cui2022xlarge_nearfar,wang2024tutorial_xlmimo}.
This coordinated rotation design brings several challenges.
Specifically, the near-field IRS-BS link couples IRS panel orientation and BS antenna boresights.
In addition, the formulated sum-rate maximization problem is highly non-convex, while efficient joint optimization methods for these variables are still lacking. 
Therefore, it is necessary to develop a joint design framework for IRS-assisted multi-user MIMO uplink systems with BS boresight steering and IRS panel rotation under the mixed-field channel model.

In this paper, we propose a coordinated dual-rotation framework for IRS-assisted MIMO uplink communications.
The main contributions are summarized as follows:

\begin{itemize}
	\item 
	We establish a coordinated dual-rotation system model for IRS-assisted multi-user MIMO uplink communications, where the BS is equipped with RAs and the IRS panel orientation is adjustable.
	Then, we formulate an uplink sum-rate maximization problem by jointly optimizing the receive beamforming, IRS phase shifts, BS antenna boresights, and IRS panel orientation, subject to unit-modulus and mechanical steering constraints.
	We study the single-user case and show that, while the dual-rotation gain is multiplicatively separable in the far field, it becomes coupled in the near field through a coherent-combining factor bounded in $[1/N,1]$. We further characterize the aperture-to-distance-ratio-dependent alignment and reflection-gain behavior.
	
	\item 
	For the general multi-user case, we develop an alternating optimization (AO) algorithm to handle the coupled variables.
	The receive beamforming is updated in closed form via the minimum mean-square error (MMSE) combiner, and the IRS phase shifts are optimized by applying fractional programming (FP) to reformulate the sum-rate objective into a quadratic form solved by Riemannian conjugate gradient over the unit-modulus manifold.
	The BS antenna boresights and IRS panel orientation are updated by projected gradient ascent (PGA) and the Barzilai-Borwein variant (BB-PGA), respectively.
	
	\item	
	Simulation results verify the convergence of the proposed algorithm.
	They also show that the proposed scheme outperforms the fixed-orientation and single-rotation benchmarks under different transmit powers, BS antenna numbers, user numbers, and IRS-BS distances.
	The decomposition results further indicate that the gain of coordinated dual rotation is mainly attributed to the improvement in the average reflected-channel power.
	Moreover, the near-field tradeoff results show that a larger normalized aperture-to-distance ratio reduces the optimized mean alignment while increasing the reflection-gain variation.

\end{itemize}

The remainder of this paper is organized as follows.
Section~\ref{sec:system} presents the system model and problem formulation.
Section~\ref{sec:nearfield} analyzes the single-user case to characterize rotation-induced near-field effects.
Section~\ref{sec:multi_user} studies the multi-user sum-rate maximization problem and the proposed AO algorithm.
Section~\ref{sec:results} presents numerical results, and Section~\ref{sec:conclusion} concludes the paper.

Boldface lower case and upper case letters denote vectors and matrices, respectively. 
$\mathbb{C}^{M \times N}$ and $\mathbb{R}^{M \times N}$ denote the sets of $M \times N$ complex and real matrices, respectively.
$\|\mathbf{x}\|$ denotes the Euclidean norm of vector $\mathbf{x}$. 
$\|\cdot\|$, $\Re\{\cdot\}$, $(\cdot)^*$, $(\cdot)^T$, and $(\cdot)^H$ denote the Euclidean norm, real part, conjugate, transpose, and Hermitian transpose, respectively.
For conformable vectors or matrices, $\odot$ denotes the Hadamard product.
$\mathrm{diag}(\cdot)$ constructs a diagonal matrix from a vector or extracts the diagonal entries of a matrix.
$\nabla_{\mathbf{x}} f(\mathbf{x})$ denotes the gradient of $f$ with respect to $\mathbf{x}$.
$\mathcal{CN}(\boldsymbol{\mu},\boldsymbol{\Sigma})$ denotes a complex Gaussian distribution.
\begin{figure}[t]
	\centering
	\includegraphics[width=\linewidth]{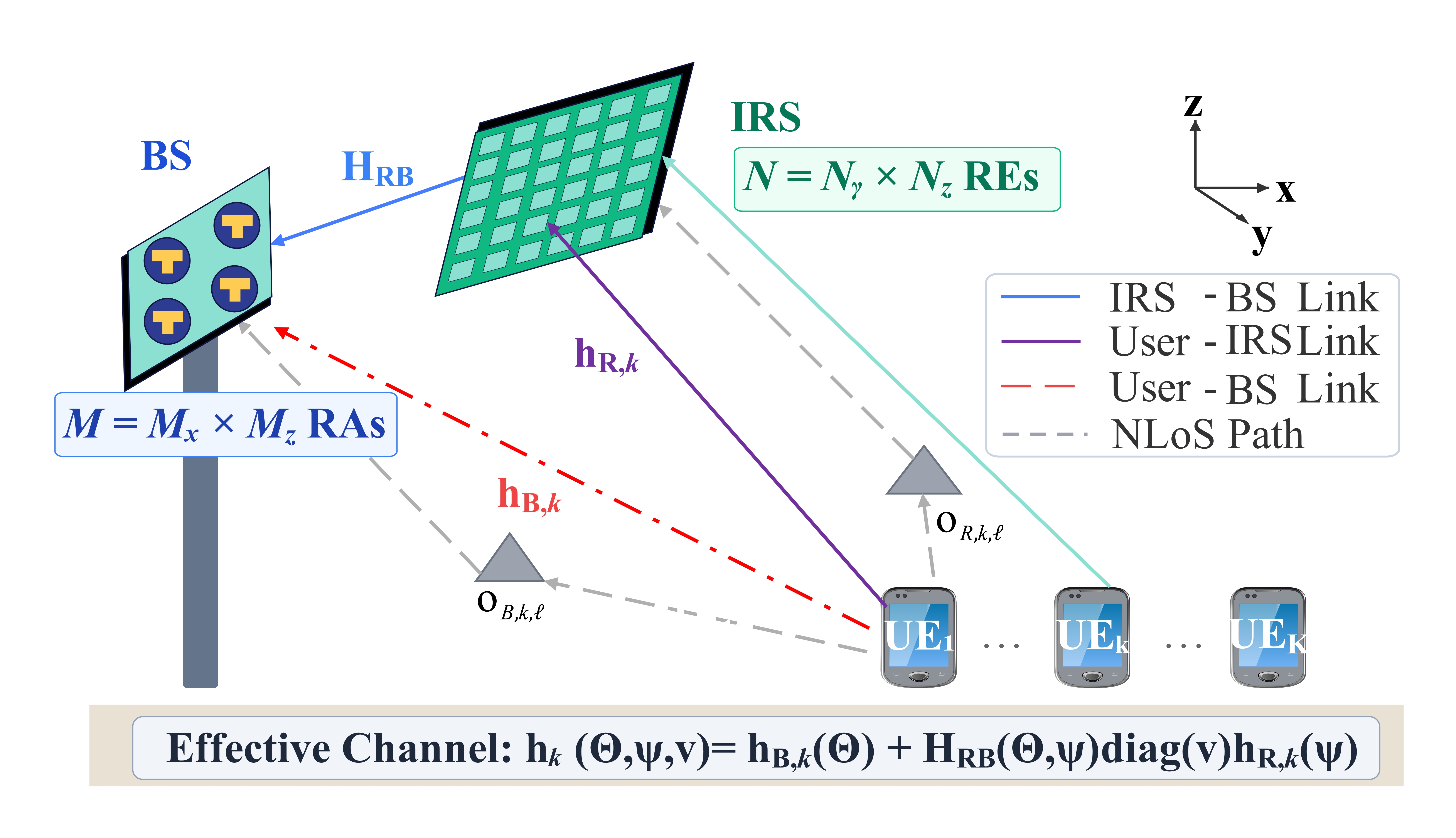}
	\caption{System model of the considered IRS-assisted multi-user uplink system.}
	\label{fig:dual_system_model}
\end{figure}

\section{System Model and Problem Formulation}
\label{sec:system}

\subsection{System Model}
\label{subsec:system_desc}

As illustrated in Fig.~\ref{fig:dual_system_model}, we consider an uplink multi-user system, where $K$ single-antenna users transmit to the BS equipped with $M=M_xM_z$ RAs, assisted by an IRS with $N=N_yN_z$ elements.
The IRS is mounted on a mechanical platform such that the panel orientation can be adjusted.
We establish a global Cartesian coordinate system, where the uniform planar array (UPA) in the BS lies on the $x$-$z$ plane with the array normal along the $+y$ axis, and elevation angles are measured from the $+y$ axis.
The position of user $k$ is denoted by $\mathbf{u}_k\in\mathbb{R}^3$.
The sets $\mathcal K \triangleq \{1,\ldots,K\}$, $\mathcal M \triangleq \{1,\ldots,M\}$, and $\mathcal N \triangleq \{1,\ldots,N\}$ collect the user, BS antenna, and IRS-element indices, respectively.

The array center is located at $\mathbf{b}_0=[0,0,z_{\mathrm{B}}]^{T}$, and the phase center of RA~$m$ is given by $\mathbf{b}_m=\mathbf{b}_0+\bar{\mathbf{b}}_m$ for all $m\in\mathcal{M}$.
Each RA adjusts the boresight direction through elevation angle $\theta_m^{\mathrm{e}}$ and azimuth angle $\theta_m^{\mathrm{a}}\in[0,2\pi)$.
Let $\boldsymbol{\theta}_m=(\theta_m^{\mathrm{e}},\theta_m^{\mathrm{a}})^{T}\in\mathbb{R}^{2}$ denote the deflection angles of RA~$m$, and define the BS antenna boresight matrix as
$\boldsymbol{\Theta}=[\boldsymbol{\theta}_1,\ldots,\boldsymbol{\theta}_M]\in\mathbb{R}^{2\times M}$.
The unit-norm boresight vector of RA~$m$ is given by
\setlength\abovedisplayskip{4pt}
\setlength\belowdisplayskip{4pt}
\begin{equation}
	\label{eq:boresight_vector}
	\mathbf{f}_m(\boldsymbol{\theta}_m)
	=
	\left[
	\sin\theta_m^{\mathrm{e}}\cos\theta_m^{\mathrm{a}},\,
	\cos\theta_m^{\mathrm{e}},\,
	\sin\theta_m^{\mathrm{e}}\sin\theta_m^{\mathrm{a}}
	\right]^{T}.
\end{equation}
When $\theta_m^{\mathrm{e}}=0$, the boresight vector reduces to the reference direction $\mathbf{f}_{\mathrm{ref}}=[0,1,0]^{T}$, which coincides with the array normal.
Boresight steering adjusts the direction-dependent gain of each RA while keeping the phase center fixed.
Mechanical steering imposes the elevation-angle constraint 
\begin{equation}
	0\leq \theta_m^{\mathrm{e}} \le \theta_{\max}, \quad \forall m\in\mathcal{M},
\end{equation}
where $\theta_{\max}\leq \pi/2$ denotes the maximum allowable elevation tilt of the mechanical platform.
By contrast, the azimuth angle $\theta_m^{\mathrm{a}}$ remains unconstrained over $[0,2\pi)$.

The IRS center is located at $\mathbf{r}_0\in\mathbb{R}^{3}$.
We adopt a square aperture with $N_y=N_z$, where $d_{\mathrm{IRS}}$ denotes the center-to-center spacing between adjacent reflecting elements.
The side length and aperture diagonal are given by $L_{\mathrm R}\triangleq (N_{\mathrm{side}}-1)d_{\mathrm{IRS}}$ and $D_{\mathrm R}\triangleq L_{\mathrm R}\sqrt{2}$, respectively.
The reference normal $\mathbf{n}_{0}$ is defined as the unit vector pointing from IRS center $\mathbf{r}_0$ toward BS array center $\mathbf{b}_0$.
The reference offset of element $n$ from the IRS center is denoted by $\bar{\mathbf r}_n$, expressed in the global coordinate system under the reference IRS orientation, with $\mathbf n_0^T\bar{\mathbf r}_n=0$.
The IRS orientation is parameterized by the Euler-angle vector $\boldsymbol{\psi}=[\alpha,\beta,\phi]^{T}$ through the active rotation
$\mathbf{R}(\boldsymbol{\psi})=\mathbf{R}_x(\phi)\mathbf{R}_y(\beta)\mathbf{R}_z(\alpha)$, where $\mathbf{R}_x(\cdot)$, $\mathbf{R}_y(\cdot)$, and $\mathbf{R}_z(\cdot)$ denote the elementary rotation matrices about the global $x$-, $y$-, and $z$-axes, respectively.
The composite matrix $\mathbf{R}(\boldsymbol{\psi})$ corresponds to the sequential application of $\mathbf{R}_z(\alpha)$, $\mathbf{R}_y(\beta)$, and $\mathbf{R}_x(\phi)$.
The rotation matrix is given by
\begin{equation}
	\label{eq:rotation_matrix}
	\!\!\!\!	\mathbf{R}(\boldsymbol{\psi})
	\!\!	=\!\!
	\begin{bmatrix}
		c_\alpha c_\beta &
		-s_\alpha c_\beta & s_\beta \\
		c_\alpha s_\beta s_\phi + s_\alpha c_\phi &
		-s_\alpha s_\beta s_\phi + c_\alpha c_\phi &
		-c_\beta s_\phi \\
		-c_\alpha s_\beta c_\phi + s_\alpha s_\phi &
		s_\alpha s_\beta c_\phi + c_\alpha s_\phi &
		c_\beta c_\phi
	\end{bmatrix},
\end{equation}
where $c_\alpha=\cos\alpha$, $s_\alpha=\sin\alpha$, and similarly for $\beta$ and $\phi$.
The identity rotation $\boldsymbol{\psi}=\mathbf{0}$ corresponds to the reference orientation.
The Euler angles satisfy $|\alpha|\leq\alpha_{\max}$, $|\beta|\leq\beta_{\max}$, and $|\phi|\leq\phi_{\max}$, where $\alpha_{\max}$, $\beta_{\max}$, and $\phi_{\max}$ denote the allowable rotation ranges of the IRS platform along the three axes.
Under rotation $\boldsymbol{\psi}$, the global  position of IRS element $n$ becomes
\begin{equation}
	\label{eq:IRS_rotation_pos}
	\mathbf{r}_n(\boldsymbol{\psi})
	= \mathbf{r}_0+\mathbf{R}(\boldsymbol{\psi})\bar{\mathbf{r}}_n,
\end{equation}
and the corresponding surface normal becomes
\begin{equation}
	\label{eq:IRS_rotation_normal}
	\mathbf{n}(\boldsymbol{\psi})
	= \mathbf{R}(\boldsymbol{\psi})\mathbf{n}_{0}.
\end{equation}

\subsection{Channel Model}
\label{subsec:channel}
The user-BS and user-IRS channels are modeled under far-field plane-wave assumptions, whereas the IRS-BS channel is modeled in the near field due to the short IRS-BS distance.\footnote{Rayleigh distance for an aperture of diagonal $D_{\mathrm{R}}$ is $R_{\mathrm{R}} = 2D_{\mathrm{R}}^2/\lambda$. The default IRS configuration in Section~\ref{sec:results} satisfies $d_{\mathrm{RB}}<R_{\mathrm{R}}$, so the IRS-BS link is modeled by a spherical-wave near-field channel.}
This leads to a mixed-field channel model with far-field user-related links and near-field IRS-BS propagation.

\subsubsection{Directional Gain Model}
Each RA and IRS element follows the cosine-power directional gain pattern~\cite{peng2025ra_spectrum}.
For a directional axis $\mathbf{n}$ and a propagation direction $\mathbf{u}$, the angular offset is $\varepsilon=\arccos(\mathbf{n}^{T}\mathbf{u})$, and the element gain is modeled as
\begin{equation}
	\label{eq:gain_pattern_modeling}
	G_{\mathrm{e}}(\varepsilon)=
	\begin{cases}
		G_0\cos^{2p}(\varepsilon), & \varepsilon\in[0,\pi/2],\\
		0, & \varepsilon\in(\pi/2,\pi],
	\end{cases}
\end{equation}
where $p\geq1/2$ controls the beamwidth, and $G_0=2(2p+1)$ ensures power conservation over the front hemisphere. 
The BS RAs and IRS elements use separate parameters $(G_{0,\mathrm{BS}},p_{\mathrm{BS}})$ and $(G_{0,\mathrm{IRS}},p_{\mathrm{IRS}})$, respectively. 
All direction vectors in $G_{\mathrm{e}}(\cdot)$ are defined along the wave-propagation direction.
Accordingly, the angular arguments of $G_{\mathrm{e}}(\cdot)$ for the user-BS reception, user-IRS incidence, IRS-BS reflection, and IRS-BS reception are given by
$\arccos(-\mathbf{f}_m^{T}\mathbf{q}_{\mathrm{B},k,\ell})$,
$\arccos(-\mathbf{n}(\boldsymbol{\psi})^{T}\mathbf{q}_{\mathrm{R},k,\ell})$,
$\arccos(\mathbf{n}(\boldsymbol{\psi})^{T}\hat{\mathbf{d}}_{n,m})$,
and
$\arccos(-\mathbf{f}_m^{T}\hat{\mathbf{d}}_{n,m})$,
respectively.

\subsubsection{User-to-BS Channel}
The direct user-BS channel $\mathbf{h}_{\mathrm{B},k}(\boldsymbol{\Theta}) \in\mathbb{C}^{M\times 1}$ from user $k$ to the BS contains $L_{\mathrm{B},k}$ multipath components, following the
geometry-based sparse channel representation~\cite{ayach2014spatially_sparse}.
The scatterer associated with component $\ell$  is located at  $\mathbf{o}_{\mathrm{B},k,\ell}\in\mathbb{R}^3$.
The corresponding propagation distance is $d_{\mathrm{B},k,\ell} = \|\mathbf{o}_{\mathrm{B},k,\ell} - \mathbf{u}_k\| + \|\mathbf{b}_0 - \mathbf{o}_{\mathrm{B},k,\ell}\|$, and the corresponding unit arrival direction at the BS is $\mathbf{q}_{\mathrm{B},k,\ell} = (\mathbf{b}_0 - \mathbf{o}_{\mathrm{B},k,\ell})/\|\mathbf{b}_0 - \mathbf{o}_{\mathrm{B},k,\ell}\|$.
For the line-of-sight (LoS) component with $\ell=1$, the propagation distance and arrival direction reduce to $d_{\mathrm{B},k,1} = \|\mathbf{b}_0 - \mathbf{u}_k\|$ and
$\mathbf{q}_{\mathrm{B},k,1} = (\mathbf{b}_0 - \mathbf{u}_k)/d_{\mathrm{B},k,1}$.
The channel vector $\mathbf{h}_{\mathrm{B},k}(\boldsymbol{\Theta})$ is given by
\begin{equation}
	\label{eq:hkB_total}
	\mathbf{h}_{\mathrm{B},k}(\boldsymbol{\Theta}) =
	\sum\nolimits_{\ell=1}^{L_{\mathrm{B},k}}
	c_{\mathrm{B},k,\ell}
	\mathrm{diag}\!\Bigl(\sqrt{
		\mathbf{g}_{\mathrm{B},k,\ell}(\boldsymbol{\Theta})}\Bigr)
	\mathbf{a}_{\mathrm{B},k,\ell},
\end{equation}
where $c_{\mathrm{B},k,\ell}\in\mathbb{C}$ denotes the complex path gain, $\mathbf{g}_{\mathrm{B},k,\ell}(\boldsymbol{\Theta}) \in\mathbb{R}^{M\times1}$ denotes the per-antenna directional gain vector, and $\mathbf{a}_{\mathrm{B},k,\ell} 	\in\mathbb{C}^{M\times1}$ denotes the BS steering vector.
For the LoS component, path gain follows the free-space model $c_{\mathrm{B},k,1} =
(\lambda/4\pi d_{\mathrm{B},k,1})\, e^{-jk_c d_{\mathrm{B},k,1}}$.
The $m$-th entry of $\mathbf{g}_{\mathrm{B},k,1}$ is $G_{\mathrm{e}}(\varepsilon_{k,m}^{\mathrm{dir}})$ with $(G_0, p) = (G_{0,\mathrm{BS}}, p_{\mathrm{BS}})$, where $\varepsilon_{k,m}^{\mathrm{dir}} = \arccos(-\mathbf{f}_m^{T}\mathbf{q}_{\mathrm{B},k,1})$ denotes the angular offset between RA boresight $\mathbf{f}_m$ and arrival direction $\mathbf{q}_{\mathrm{B},k,1}$. The corresponding LoS steering vector entry is
\begin{equation}
	\label{eq:BS_LoS_steering}
	[\mathbf{a}_{\mathrm{B},k,1}]_m = \exp \bigl(-jk_c\,\mathbf{q}_{\mathrm{B},k,1}^{T}(\mathbf{b}_m - \mathbf{b}_0)\bigr).
\end{equation}
For non-LoS (NLoS) components with $\ell\ge2$, the complex path gain is modeled as $c_{\mathrm{B},k,\ell}=\kappa_{\mathrm{NLoS}}\frac{\lambda}{4\pi d_{\mathrm{B},k,\ell}}e^{-jk_c d_{\mathrm{B},k,\ell}+j\varphi_{\mathrm{B},k,\ell}},$
where $\kappa_{\mathrm{NLoS}}\in(0,1)$ denotes the NLoS amplitude scaling factor, and $\varphi_{\mathrm{B},k,\ell}$ is uniformly distributed over $[0,2\pi)$.
For notational consistency, the directional gain vector and steering vector for each NLoS path follow the same form as the LoS component in~\eqref{eq:BS_LoS_steering}, with $\mathbf{q}_{\mathrm{B},k,1}$ replaced by $\mathbf{q}_{\mathrm{B},k,\ell}$.

\subsubsection{User-to-IRS Channel}
The user-to-IRS channel $\mathbf{h}_{\mathrm{R},k}(\boldsymbol{\psi}) \in\mathbb{C}^{N\times1}$ from user $k$ to the IRS contains $L_{\mathrm{IR},k}$ multipath components.
The scatterer associated with component $\ell$ is located at $\mathbf{o}_{\mathrm{R},k,\ell}\in\mathbb{R}^3$. The corresponding propagation distance is $d_{\mathrm{R},k,\ell} = \|\mathbf{o}_{\mathrm{R},k,\ell} - \mathbf{u}_k\| + \|\mathbf{r}_0 - \mathbf{o}_{\mathrm{R},k,\ell}\|$, and the corresponding unit arrival direction at the IRS center is $\mathbf{q}_{\mathrm{R},k,\ell} = (\mathbf{r}_0 -  \mathbf{o}_{\mathrm{R},k,\ell})/ \|\mathbf{r}_0 - \mathbf{o}_{\mathrm{R},k,\ell}\|$.
For the LoS component with $\ell=1$, the propagation distance and arrival direction reduce to
$d_{\mathrm{R},k,1} = \|\mathbf{r}_0 - \mathbf{u}_k\|$ and $\mathbf{q}_{\mathrm{R},k,1} =
(\mathbf{r}_0 - \mathbf{u}_k)/d_{\mathrm{R},k,1}$.
The user-to-IRS channel vector is given by
\begin{equation}
	\label{eq:hkR_multipath}
	\mathbf{h}_{\mathrm{R},k}(\boldsymbol{\psi})
	=
	\sum_{\ell=1}^{L_{\mathrm{IR},k}}
	c_{\mathrm{R},k,\ell}\,
	\sqrt{G_{k,\ell}^{\mathrm{inc}}(\boldsymbol{\psi})}\,
	\mathbf{a}_{\mathrm{R},k,\ell}(\boldsymbol{\psi}),
\end{equation}
where $c_{\mathrm{R},k,\ell}\in\mathbb{C}$ denotes the complex path gain, $G_{k,\ell}^{\mathrm{inc}}(\boldsymbol{\psi})$ denotes the incidence-direction gain, and
$\mathbf{a}_{\mathrm{R},k,\ell}(\boldsymbol{\psi}) \in\mathbb{C}^{N\times1}$ denotes the IRS steering vector. 
For the LoS component, path gain follows the free-space model
$c_{\mathrm{R},k,1} = (\lambda/4\pi d_{\mathrm{R},k,1})\, e^{-jk_c d_{\mathrm{R},k,1}}$.
The incidence-direction gain is $G_{k,1}^{\mathrm{inc}}(\boldsymbol{\psi}) =
G_{\mathrm{e}}(\varepsilon_k^{\mathrm{inc}})$, where $\varepsilon_k^{\mathrm{inc}} =
\arccos(-\mathbf{n}(\boldsymbol{\psi})^{T} \mathbf{q}_{\mathrm{R},k,1})$ denotes the incidence angle between surface normal $\mathbf{n}(\boldsymbol{\psi})$ and arrival direction $\mathbf{q}_{\mathrm{R},k,1}$, with gain parameters $(G_0,p) = (G_{0,\mathrm{IRS}}, p_{\mathrm{IRS}})$. 
The $n$-th entry of the LoS steering vector is
\begin{equation}
	\label{eq:IRS_LoS_steering}
	[\mathbf{a}_{\mathrm{R},k,1}(\boldsymbol{\psi})]_n = \exp \bigl(-jk_c\,\mathbf{q}_{\mathrm{R},k,1}^{T}(\mathbf{r}_n(\boldsymbol{\psi}) - \mathbf{r}_0)\bigr).
\end{equation}
For NLoS components with $\ell\ge2$, the complex path gain is modeled as
$c_{\mathrm{R},k,\ell}=\kappa_{\mathrm{NLoS}}\frac{\lambda}{4\pi d_{\mathrm{R},k,\ell}}e^{-jk_c d_{\mathrm{R},k,\ell}+j\varphi_{\mathrm{R},k,\ell}}
$,	where $\varphi_{\mathrm{R},k,\ell}$ is uniformly distributed over $[0,2\pi)$.
The incidence-direction gain and steering vector for each NLoS path follow the same form as the LoS component in~\eqref{eq:IRS_LoS_steering}, with $\mathbf{q}_{\mathrm{R},k,1}$ replaced by $\mathbf{q}_{\mathrm{R},k,\ell}$.

\subsubsection{IRS-to-BS Channel}
The IRS-BS channel
$\mathbf{H}_{\mathrm{RB}}(\boldsymbol{\Theta},\boldsymbol{\psi}) \in\mathbb{C}^{M\times N}$ captures element-wise near-field propagation between the IRS and the BS.
We model the IRS-BS channel with a single LoS component.
For each element-antenna pair $(n,m)$, the propagation distance is $d_{n,m}(\boldsymbol{\psi}) = \|\mathbf{b}_m - \mathbf{r}_n(\boldsymbol{\psi})\|$, and the unit direction vector from IRS element $n$ toward BS antenna $m$ is $\hat{\mathbf{d}}_{n,m}(\boldsymbol{\psi})
= (\mathbf{b}_m - \mathbf{r}_n(\boldsymbol{\psi}))/ d_{n,m}(\boldsymbol{\psi})$.
The $(m,n)$-th entry of the IRS-BS channel matrix is
\begin{equation}
	\label{eq:HRB_entry}
	\!\!\!	[\mathbf{H}_{\mathrm{RB}}]_{m,n}
	\!=\!
	\frac{\lambda\sqrt{G_{n,m}^{\mathrm{ref}}(\boldsymbol{\psi})\,
			G_{m,n}^{\mathrm{BS}}(\boldsymbol{\theta}_m,\boldsymbol{\psi})}}
	{4\pi d_{n,m}(\boldsymbol{\psi})}
	\,e^{-jk_c d_{n,m}(\boldsymbol{\psi})},
\end{equation}
where $G_{n,m}^{\mathrm{ref}}(\boldsymbol{\psi}) = G_{\mathrm{e}}(\varepsilon_{n,m}^{\mathrm{ref}})$ denotes the IRS reflection-direction gain at element $n$ toward antenna $m$, and $G_{m,n}^{\mathrm{BS}}(\boldsymbol{\theta}_m,\boldsymbol{\psi}) = G_{\mathrm{e}}(\varepsilon_{m,n}^{\mathrm{BS}})$ denotes the BS receive directional gain at antenna $m$ from element $n$. 
The IRS reflection angle and the BS receive off-boresight angle are defined as
$\varepsilon_{n,m}^{\mathrm{ref}} = \arccos(\mathbf{n}^{T}\hat{\mathbf{d}}_{n,m})$
and
$\varepsilon_{m,n}^{\mathrm{BS}} = \arccos(-\mathbf{f}_m^{T}\hat{\mathbf{d}}_{n,m})$,
respectively, where they characterize the angular mismatch with respect to the IRS surface normal and the RA boresight.

\subsection{Signal Model and Problem Formulation}
\label{subsec:problem}
The IRS phase-shift matrix is denoted by
$\mathbf{V}=\mathrm{diag}(\mathbf{v})\in\mathbb{C}^{N\times N}$, where
$\mathbf{v}=[v_1,\ldots,v_N]^T$ and
$v_n=e^{j\phi_n}$ with $\phi_n\in[0,2\pi)$, $\forall n\in\mathcal{N}$.
Each user $k$ transmits symbol $x_k\sim\mathcal{CN}(0,1)$ with uplink power $P_k\ge 0$.
The effective  channel of user $k$ combines direct and IRS-reflected paths as
\begin{equation}
	\label{eq:composite_channel}
	\!\!\!	\mathbf{h}_k(\boldsymbol{\Theta},\boldsymbol{\psi},\mathbf{v})
	\!	=\!
	\mathbf{h}_{\mathrm{B},k}(\boldsymbol{\Theta})
	\!	+\!
	\mathbf{H}_{\mathrm{RB}}(\boldsymbol{\Theta},\boldsymbol{\psi})
	\,\mathrm{diag}(\mathbf{v})\,
	\mathbf{h}_{\mathrm{R},k}(\boldsymbol{\psi}).
\end{equation}
The received signal at the BS is
\begin{equation}
	\label{eq:received_signal}
	\mathbf{y}
	=
	\sum\nolimits_{k=1}^{K}\sqrt{P_k}\,
	\mathbf{h}_k(\boldsymbol{\Theta},\boldsymbol{\psi},\mathbf{v})
	x_k + \mathbf{z},
\end{equation}
where $\mathbf{z}\sim\mathcal{CN}(\mathbf{0},\sigma^2\mathbf{I}_M)$ denotes additive white Gaussian noise with power $\sigma^2$, and $\mathbf{I}_M$ denotes the $M\times M$ identity matrix. 
A nonzero linear receive beamforming vector $\mathbf w_k\in\mathbb C^{M\times 1}$ is applied at the BS for user $k$, $\mathbf W=[\mathbf w_1,\ldots,\mathbf w_K]\in\mathbb C^{M\times K}$.
The signal-to-interference-plus-noise ratio (SINR) of user $k$ is
\begin{equation}
	\label{eq:sinr}
	\gamma_k
	=
	\frac{P_k|\mathbf{w}_k^{H}\mathbf{h}_k|^2}
	{\displaystyle\sum\nolimits_{j\neq k}
		P_j|\mathbf{w}_k^{H}\mathbf{h}_j|^2
		+\sigma^2\|\mathbf{w}_k\|^2}.
\end{equation}
The achievable uplink sum rate is $R_{\mathrm{sum}} = \sum\nolimits_{k=1}^{K}\log_2(1+\gamma_k)$.

We aim to maximize the sum rate by jointly optimizing the receive beamforming matrix $\mathbf{W}$, the BS boresight-angle matrix $\boldsymbol{\Theta}$, the IRS orientation vector $\boldsymbol{\psi}$, and the IRS reflection vector $\mathbf{v}$.
The optimization problem is formulated as
\begin{subequations}
	\label{eq:P1}
	\begin{align}
		(\mathrm{P1}):\ \max_{\mathbf{W},\boldsymbol{\Theta},
			\boldsymbol{\psi},\mathbf{v}}\
		& \sum_{k=1}^{K}\log_2(1+\gamma_k)
		\label{eq:P1_obj}\\
		\text{s.t.}\
		& \mathbf{n}(\boldsymbol{\psi})^{T}
		(\mathbf{b}_0-\mathbf{r}_0)\ge 0,
		\label{eq:P1_const_IRS_bs}\\
		& |\alpha|\leq\alpha_{\max},  |\beta|\leq\beta_{\max},  
		|\phi|\leq\phi_{\max},
		\label{eq:P1_const_box}\\
		& 0\le\theta_m^{\mathrm{e}}\le\theta_{\max},\
		\forall m\in\mathcal{M},
		\label{eq:P1_const_elev}\\
		& |v_n|=1,\ \forall n\in\mathcal{N}.
		\label{eq:P1_const_unimodular}
	\end{align}
\end{subequations}
Constraints~\eqref{eq:P1_const_IRS_bs}-\eqref{eq:P1_const_unimodular} represent the IRS half-space visibility condition, the allowable ranges of the three IRS Euler angles, the RA elevation steering limits, and the unit-modulus IRS reflection coefficients, respectively.
Problem~(P1) is non-convex and difficult to solve optimally due to the coupled optimization variables in the objective function and the unit-modulus IRS reflection constraints.

\section{Single-User System} 
\label{sec:nearfield}

In this section, we specialize the proposed uplink design to the single-user case with $K=1$.
Problem~(P1) is rewritten as
\begin{subequations}
	\label{eq:PSU}
	\begin{align}
		(\mathrm{P2}):\
		\max_{\mathbf{w},\boldsymbol{\Theta},\boldsymbol{\psi},\mathbf{v}}\
		&
		\frac{P\,|\mathbf{w}^{H}\mathbf{h}(\boldsymbol{\Theta},\boldsymbol{\psi},\mathbf{v})|^{2}}
		{\sigma^{2}\|\mathbf{w}\|^{2}}
		\label{eq:PSU_obj}\\
		\mathrm{s.t.}\
		&
		\eqref{eq:P1_const_IRS_bs},
		\eqref{eq:P1_const_box},
		\eqref{eq:P1_const_elev},
		\eqref{eq:P1_const_unimodular},
	\end{align}
\end{subequations}
where $P$ denotes the uplink transmit power and $\mathbf{w}\in\mathbb{C}^{M\times 1}$ is the receive combining vector.
For any given $\boldsymbol{\Theta}$, $\boldsymbol{\psi}$, and $\mathbf v$, 
the effective channel $\mathbf h(\boldsymbol{\Theta},\boldsymbol{\psi},\mathbf v)$ is fixed. 
By the Cauchy-Schwarz inequality, the receive-combining subproblem is maximized by the maximum-ratio combining (MRC) receiver~\cite{tse2005fundamentals}. 
Thus, without loss of optimality, the unit-norm receive combiner is given by
\begin{equation}
	\mathbf w^{\star}
	=
	\frac{\mathbf h(\boldsymbol{\Theta},\boldsymbol{\psi},\mathbf v)}
	{\|\mathbf h(\boldsymbol{\Theta},\boldsymbol{\psi},\mathbf v)\|}.
\end{equation}
The corresponding maximum signal-to-noise ratio (SNR) is 
\begin{equation}
	\gamma^{\star}=P\|\mathbf{h}(\boldsymbol{\Theta},\boldsymbol{\psi},\mathbf{v})\|^{2}/\sigma^{2}.
\end{equation}
With the MRC receiver, Problem~(P2) can be equivalently rewritten as
\begin{subequations}
	\label{eq:PSU_equiv}
	\begin{align}
		\label{eq:PSU_eq}
		\max_{\boldsymbol{\Theta},\boldsymbol{\psi},\mathbf{v}}\
		&F(\boldsymbol{\Theta},\boldsymbol{\psi},\mathbf{v})\\
		\quad\mathrm{s.t.}\
		&
		\eqref{eq:P1_const_IRS_bs},\eqref{eq:P1_const_box},\eqref{eq:P1_const_elev},\eqref{eq:P1_const_unimodular},
	\end{align}
\end{subequations}
where $F(\boldsymbol{\Theta},\boldsymbol{\psi},\mathbf{v})=
\|\mathbf{h}(\boldsymbol{\Theta},\boldsymbol{\psi},\mathbf{v})\|^{2}$.
The problem~\eqref{eq:PSU_equiv} remains non-convex due to the unit-modulus constraints on $\mathbf{v}$ and the rotation-dependent cascaded channel.

\subsection{Proposed Solution}
\label{subsec:su_AO}

For problem~\eqref{eq:PSU_equiv}, we update $\mathbf{v}$, $\boldsymbol{\Theta}$, and $\boldsymbol{\psi}$ alternately with the other variables fixed.

\emph{1) IRS Phase-Shift Update:}
For fixed $\boldsymbol{\Theta}$ and $\boldsymbol{\psi}$, let
$\mathbf{h}_{\mathrm{RB},n}(\boldsymbol{\Theta},\boldsymbol{\psi})$
denote the $n$-th column of
$\mathbf{H}_{\mathrm{RB}}(\boldsymbol{\Theta},\boldsymbol{\psi})$.
Define $\mathbf{A}=[\mathbf{a}_1,\ldots,\mathbf{a}_n,\ldots,\mathbf{a}_N]$, where
\begin{equation}
	\label{eq:su_an_def}
	\mathbf{a}_n
	=
	[\mathbf{h}_{\mathrm{R}}(\boldsymbol{\psi})]_n
	\mathbf{h}_{\mathrm{RB},n}(\boldsymbol{\Theta},\boldsymbol{\psi}),
	\quad n=1,\ldots,N .
\end{equation}
Then
$\mathbf{H}_{\mathrm{RB}}(\boldsymbol{\Theta},\boldsymbol{\psi})
\mathrm{diag}(\mathbf{v})\mathbf{h}_{\mathrm{R}}(\boldsymbol{\psi})
=\mathbf{A}\mathbf{v}$, and the phase-shift subproblem becomes
\begin{equation}
	\label{eq:su_phase_subproblem}
	\begin{aligned}
		\max_{\mathbf{v}}\quad
		&
		\left\|
		\mathbf{h}_{\mathrm{B}}(\boldsymbol{\Theta})
		+
		\mathbf{A}\mathbf{v}
		\right\|^{2}\\
		\mathrm{s.t.}\quad
		&
		|v_n|=1,\quad n=1,\ldots,N .
	\end{aligned}
\end{equation}
We update $\mathbf{v}$ via block coordinate descent (BCD), where each phase shift is updated with the other phase shifts fixed.
With $\{v_i\}_{i\neq n}$ fixed, the per-coordinate subproblem is
\begin{equation}
	\label{eq:su_vn_bcd}
	\begin{aligned}
		\max_{v_n}\quad
		&
		|v_n|^{2}\mathbf{a}_n^{H}\mathbf{a}_n
		+
		2\Re\!\left\{
		v_n^{*}\eta_n
		\right\}\\
		\mathrm{s.t.}\quad
		&
		|v_n|=1 ,
	\end{aligned}
\end{equation}
where
$\eta_n
\triangleq
\mathbf{a}_n^{H} \left(
\mathbf{h}_{\mathrm{B}}(\boldsymbol{\Theta})
+\sum_{i\neq n}v_i\mathbf{a}_i
\right)$.
The corresponding closed-form update is
\begin{equation}
	\label{eq:su_phase_update}
	v_n^{\mathrm{opt}}
	=
	e^{j\angle\eta_n},
\end{equation}
where any unit-modulus value is feasible when $\eta_n=0$.

\emph{2) Rotation Update:}
For fixed $\mathbf{v}$, the rotation subproblem is given by
\begin{equation}
	\label{eq:su_rotation_subproblem}
	\begin{aligned}
		\max_{\boldsymbol{\Theta},\boldsymbol{\psi}}\quad
		&
		F(\boldsymbol{\Theta},\boldsymbol{\psi},\mathbf{v})\\
		\mathrm{s.t.}\quad
		&
		\eqref{eq:P1_const_IRS_bs},
		\eqref{eq:P1_const_box},
		\eqref{eq:P1_const_elev}.
	\end{aligned}
\end{equation}
We update the BS boresights and IRS orientation by PGA.

For the BS boresight update, $\boldsymbol{\psi}$ and
$\{\boldsymbol{\theta}_{i}\}_{i\neq m}$ are fixed.
The gradient of $F$ with respect to $\boldsymbol{\theta}_m$ is
$	\nabla_{\boldsymbol{\theta}_m}F
=
2\Re\!\left\{
[\mathbf{h}]_m^{*}
\frac{\partial[\mathbf{h}]_m}
{\partial\boldsymbol{\theta}_m}
\right\}$,
where $\partial[\mathbf{h}]_m/\partial\boldsymbol{\theta}_m$
is given by~\eqref{eq:su_RA_entry_derivative} at the top of the next page.
\begin{figure*}[!t]
	\vspace*{4pt}
	\begin{equation}
		\label{eq:su_RA_entry_derivative}
		\frac{\partial[\mathbf{h}]_m}{\partial\boldsymbol{\theta}_m}
		=
		\sum\nolimits_{\ell=1}^{L_{\mathrm{B}}}
		c_{\mathrm{B},\ell}
		[\mathbf{a}_{\mathrm{B},\ell}]_m
		\frac{\partial\sqrt{[\mathbf{g}_{\mathrm{B},\ell}]_m}}
		{\partial\boldsymbol{\theta}_m}
		+
		\sum\nolimits_{n=1}^{N}
		\frac{\lambda v_n[\mathbf{h}_{\mathrm{R}}]_n
			\sqrt{G^{\mathrm{ref}}_{n,m}(\boldsymbol{\psi})}}
		{4\pi d_{n,m}(\boldsymbol{\psi})}
		e^{-jk_c d_{n,m}(\boldsymbol{\psi})}
		\frac{\partial\sqrt{G^{\mathrm{BS}}_{m,n}}}
		{\partial\boldsymbol{\theta}_m}.
	\end{equation}
	\hrulefill
	\vspace*{4pt}
\end{figure*}
Let $\zeta_{\mathrm{BS}}=-p_{\mathrm{BS}}\sqrt{G_{0,\mathrm{BS}}}$, the square-root gain derivatives are given by
	\begin{align}
		\!\!	&\frac{\partial\sqrt{[\mathbf{g}_{\mathrm{B},\ell}]_m}}
		{\partial\boldsymbol{\theta}_m}
		\!=\!
		\zeta_{\mathrm{BS}}
		(-\mathbf{f}_m^{T}\mathbf{q}_{\mathrm{B},\ell})^{p_{\mathrm{BS}}-1}
		\left(
		\frac{\partial\mathbf{f}_m}{\partial\boldsymbol{\theta}_m}
		\right)^{T}
		\mathbf{q}_{\mathrm{B},\ell}, 
		\label{eq:su_RA_gain_derivatives}\\
		\!\!	&\frac{\partial\sqrt{G^{\mathrm{BS}}_{m,n}}}
		{\partial\boldsymbol{\theta}_m}
		\!\!=\! 
		\zeta_{\mathrm{BS}}\!\! 
		\left(
		-\mathbf{f}_m^{T}\hat{\mathbf{d}}_{n,m}(\boldsymbol{\psi})
		\right)^{p_{\mathrm{BS}}-1}\!\!\!
		\left(
		\frac{\partial\mathbf{f}_m}{\partial\boldsymbol{\theta}_m}
		\right)^{T}\!\!\!\!\!\!
		\hat{\mathbf{d}}_{n,m}(\boldsymbol{\psi}),
		\label{eq:su_RA_gain_derivatives2}
	\end{align}
where 
\begin{equation}
	\label{eq:su_boresight_jacobian}
	\frac{\partial\mathbf{f}_m}{\partial\boldsymbol{\theta}_m}
	=
	\begin{bmatrix}
		\cos\theta_m^{\mathrm{e}}\cos\theta_m^{\mathrm{a}} &
		-\sin\theta_m^{\mathrm{e}}\sin\theta_m^{\mathrm{a}}\\
		-\sin\theta_m^{\mathrm{e}} &
		0 \\
		\cos\theta_m^{\mathrm{e}}\sin\theta_m^{\mathrm{a}} &
		\sin\theta_m^{\mathrm{e}}\cos\theta_m^{\mathrm{a}}
	\end{bmatrix}.
\end{equation}
The derivatives in~\eqref{eq:su_RA_gain_derivatives} and~\eqref{eq:su_RA_gain_derivatives2} are evaluated only for positive directional cosines.
If the corresponding direction lies outside the front hemisphere, the gain term and its derivative are set to zero according to~\eqref{eq:gain_pattern_modeling}.
At iteration $r$, the boresight angle of RA~$m$ is updated by
\begin{equation}
	\label{eq:su_RA_projected_update}
	\boldsymbol{\theta}_m^{(r+1)}
	=
	\Pi_{\mathcal{X}_{\boldsymbol{\theta}_m}}
	\left(
	\boldsymbol{\theta}_m^{(r)}
	+
	s_{\mathrm{RA}}^{(r)}
	\nabla_{\boldsymbol{\theta}_m}F
	\right),
\end{equation}
where $\mathcal{X}_{\boldsymbol{\theta}_m}$ denotes the feasible set specified by~\eqref{eq:P1_const_elev}, and
$\Pi_{\mathcal{X}_{\boldsymbol{\theta}_m}}(\cdot)$ denotes the projection onto this set.
The step size $s_{\mathrm{RA}}^{(r)}$ is chosen by backtracking such that the objective value is non-decreasing after the update.

For the IRS orientation update, $\boldsymbol{\Theta}$ is fixed.
The gradient with respect to $\boldsymbol{\psi}$ is assembled from the three Euler-angle components.
For $\psi_i\in\{\alpha,\beta,\phi\}$, the $i$-th component is
\begin{equation}
	\label{eq:su_psi_gradient}
	\left[\nabla_{\boldsymbol{\psi}}F\right]_i
	=
	2\Re\!\left\{
	\left(
	\frac{\partial\mathbf{h}}{\partial\psi_i}
	\right)^{H}
	\mathbf{h}
	\right\},
\end{equation}
where
\begin{equation}
	\label{eq:su_channel_derivative_psi}
	\frac{\partial\mathbf{h}}{\partial\psi_i}
	=
	\mathbf{H}_{\mathrm{RB}}
	\mathrm{diag}(\mathbf{v})
	\frac{\partial\mathbf{h}_{\mathrm{R}}}{\partial\psi_i}
	+
	\frac{\partial\mathbf{H}_{\mathrm{RB}}}{\partial\psi_i}
	\mathrm{diag}(\mathbf{v})
	\mathbf{h}_{\mathrm{R}}.
\end{equation}
Here, $\partial[\mathbf{h}_{\mathrm{R}}]_n/\partial\psi_i$ and $ \partial[\mathbf{H}_{\mathrm{RB}}]_{m,n}/\partial\psi_i$ are given in~\eqref{eq:su_hR_derivative} and ~\eqref{eq:su_HRB_derivative}, respectively, at the top of the next page.
\begin{figure*}[!t]
	\vspace*{1pt}
	\begin{equation}
		\label{eq:su_hR_derivative}
		\begin{aligned}
			\frac{\partial[\mathbf{h}_{\mathrm{R}}]_n}{\partial\psi_i}
			=
			\sum\nolimits_{\ell=1}^{L_{\mathrm{IR}}}
			c_{\mathrm{R},\ell}
			\Bigg[
			&
			-p_{\mathrm{IRS}}\sqrt{G_{0,\mathrm{IRS}}}
			\bigl[\cos(\varepsilon_{\ell}^{\mathrm{inc}})\bigr]^{p_{\mathrm{IRS}}-1}
			\left(
			\frac{\partial\mathbf{n}}{\partial\psi_i}
			\right)^{T}
			\mathbf{q}_{\mathrm{R},\ell}
			[\mathbf{a}_{\mathrm{R},\ell}]_n\\
			&
			+
			\sqrt{G_{\ell}^{\mathrm{inc}}(\boldsymbol{\psi})}
			\left(
			-jk_c\mathbf{q}_{\mathrm{R},\ell}^{T}
			\frac{\partial\mathbf{r}_n}{\partial\psi_i}
			\right)
			e^{-jk_c\mathbf{q}_{\mathrm{R},\ell}^{T}
				(\mathbf{r}_n-\mathbf{r}_0)}
			\Bigg].
		\end{aligned}
	\end{equation}
	\hrulefill
	\vspace*{1pt}
\end{figure*}
\begin{figure*}[!t]
	\vspace*{1pt}
	\begin{equation}
		\label{eq:su_HRB_derivative}
		\begin{aligned}
			\frac{\partial[\mathbf{H}_{\mathrm{RB}}]_{m,n}}{\partial\psi_i}
			={}&
			\frac{\lambda}{4\pi}
			\frac{\partial d_{n,m}^{-1}}{\partial\psi_i}
			\sqrt{G^{\mathrm{ref}}_{n,m}G^{\mathrm{BS}}_{m,n}}
			e^{-jk_c d_{n,m}}
			+
			\frac{\lambda}{4\pi d_{n,m}}
			\frac{\partial\sqrt{G^{\mathrm{ref}}_{n,m}}}{\partial\psi_i}
			\sqrt{G^{\mathrm{BS}}_{m,n}}
			e^{-jk_c d_{n,m}}\\
			&+
			\frac{\lambda}{4\pi d_{n,m}}
			\sqrt{G^{\mathrm{ref}}_{n,m}}
			\frac{\partial\sqrt{G^{\mathrm{BS}}_{m,n}}}{\partial\psi_i}
			e^{-jk_c d_{n,m}}
			-
			jk_c
			\frac{\lambda}{4\pi d_{n,m}}
			\sqrt{G^{\mathrm{ref}}_{n,m}G^{\mathrm{BS}}_{m,n}}
			\frac{\partial d_{n,m}}{\partial\psi_i}
			e^{-jk_c d_{n,m}} .
		\end{aligned}
	\end{equation}
	\hrulefill
	\vspace*{1pt}
\end{figure*}
Let $\zeta_{\mathrm{IRS}}=p_{\mathrm{IRS}}\sqrt{G_{0,\mathrm{IRS}}}$, 
then the gain derivatives in~\eqref{eq:su_HRB_derivative} are
	\begin{align}
		\!\!\!\!	\frac{\partial\sqrt{G^{\mathrm{ref}}_{n,m}}}{\partial\psi_i}
		\!\!=\!\!{}&
		\zeta_{\mathrm{IRS}}
		\rho_{n,m}^{p_{\mathrm{IRS}}-1}
		\left[
		\left(
		\frac{\partial\mathbf{n}}{\partial\psi_i}
		\right)^{T}
		\hat{\mathbf{d}}_{n,m}
		\!\!+\!\!
		\mathbf{n}^{T}
		\frac{\partial\hat{\mathbf{d}}_{n,m}}{\partial\psi_i}
		\right],
		\label{eq:su_Gref_GBS_derivative}\\
		\!\!\!\!\frac{\partial\sqrt{G^{\mathrm{BS}}_{m,n}}}{\partial\psi_i}
		\!\!=\!\!{}&
		\zeta_{\mathrm{BS}}
		\bigl[\cos(\varepsilon^{\mathrm{BS}}_{m,n})\bigr]^{p_{\mathrm{BS}}-1}
		\mathbf{f}_m^{T}
		\frac{\partial\hat{\mathbf{d}}_{n,m}}{\partial\psi_i},
		\label{eq:su_Gref_GBS_derivative2}
	\end{align}
where $\rho_{n,m}=\mathbf{n}^{T}\hat{\mathbf{d}}_{n,m}$.
The distance and direction derivatives in~\eqref{eq:su_HRB_derivative} , \eqref{eq:su_Gref_GBS_derivative} and \eqref{eq:su_Gref_GBS_derivative2} are
\begin{equation}
	\label{eq:su_d_derivative}
	\begin{aligned}
		\!\!\!\!\!\!	\frac{\partial d_{n,m}}{\partial\psi_i}
		\!\!=\!\!
		-\hat{\mathbf{d}}_{n,m}^{T}
		\frac{\partial\mathbf{r}_n}{\partial\psi_i},\ 
		\frac{\partial\hat{\mathbf{d}}_{n,m}}{\partial\psi_i}
		\!\!=\!\!
		-\frac{\mathbf{I}_3-\hat{\mathbf{d}}_{n,m}\hat{\mathbf{d}}_{n,m}^{T}}
		{d_{n,m}}
		\frac{\partial\mathbf{r}_n}{\partial\psi_i},
	\end{aligned}
\end{equation}
where
$\partial\mathbf{r}_n/\partial\psi_i
=
(\partial\mathbf{R}/\partial\psi_i)\bar{\mathbf{r}}_n$,
$\partial\mathbf{n}/\partial\psi_i
=
(\partial\mathbf{R}/\partial\psi_i)\mathbf{n}_0$, and
$\partial d_{n,m}^{-1}/\partial\psi_i
=
-d_{n,m}^{-2}\partial d_{n,m}/\partial\psi_i$. 
Define
$g_{\mathrm{vis}}(\boldsymbol{\psi})
=
\mathbf{n}(\boldsymbol{\psi})^{T}(\mathbf{b}_0-\mathbf{r}_0)$.
Given the local point $\boldsymbol{\psi}^{(r)}$, the first-order approximation of $g_{\mathrm{vis}}(\boldsymbol{\psi})$ is
\begin{equation}
	\label{eq:su_facing_linearization}
	\widehat{g}_{\mathrm{vis}}^{(r)}(\boldsymbol{\psi})
	=
	g_{\mathrm{vis}}(\boldsymbol{\psi}^{(r)})
	+
	\nabla_{\boldsymbol{\psi}}g_{\mathrm{vis}}(\boldsymbol{\psi}^{(r)})^{T}
	\left(
	\boldsymbol{\psi}
	-
	\boldsymbol{\psi}^{(r)}
	\right).
\end{equation}
Then the local feasible set for the IRS orientation update is
\begin{equation}
	\label{eq:su_local_psi_set}
	\mathcal{C}_{\boldsymbol{\psi}}^{(r)}
	=
	\left\{
	\boldsymbol{\psi}\in\mathcal{X}_{\boldsymbol{\psi}}
	\mid
	\widehat{g}_{\mathrm{vis}}^{(r)}(\boldsymbol{\psi})\ge 0
	\right\}, 
\end{equation}
where $\mathcal{X}_{\boldsymbol{\psi}}$ denotes the Euler-angle box set specified by~\eqref{eq:P1_const_box}.
At iteration $r$, the IRS orientation is updated by
\begin{equation}
	\label{eq:su_IRS_projected_update}
	\bar{\boldsymbol{\psi}}^{(r+1)}
	=
	\Pi_{\mathcal{C}_{\boldsymbol{\psi}}^{(r)}}
	\left(
	\boldsymbol{\psi}^{(r)}
	+
	s_{\mathrm{IRS}}^{(r)}
	\nabla_{\boldsymbol{\psi}}F
	\right), 
\end{equation}
where $\Pi_{\mathcal{C}_{\boldsymbol{\psi}}^{(r)}}(\cdot)$ denotes the Euclidean projection onto
$\mathcal{C}_{\boldsymbol{\psi}}^{(r)}$.
The step size $s_{\mathrm{IRS}}^{(r)}$ is reduced by backtracking until
$\bar{\boldsymbol{\psi}}^{(r+1)}$ satisfies the original facing constraint in~\eqref{eq:P1_const_IRS_bs} and gives a non-decreasing value of $F$.
Then, we set
$\boldsymbol{\psi}^{(r+1)}=\bar{\boldsymbol{\psi}}^{(r+1)}$.

The single-user AO convergence follows from two facts.
First, the IRS phase update and the accepted rotation updates ensure that the objective value is non-decreasing over AO iterations.
Second, $F(\boldsymbol{\Theta},\boldsymbol{\psi},\mathbf v)$ is upper-bounded due to bounded rotation ranges, finite directional gains, nonzero propagation distances, and unit-modulus IRS coefficients.
Therefore, the objective sequence of the single-user AO procedure converges.

\subsection{Power Scaling and Rotation-Gain Mechanism}
\label{subsec:su_mechanism}

To gain useful insights into the dual-rotation mechanism, we first consider the far-field IRS-BS channel. 
Under the far-field, let
$\hat{\mathbf d}_0=(\mathbf b_0-\mathbf r_0)/d_{\mathrm{RB}}$
denote the center direction from the IRS to the BS.
Define
$\varepsilon_{m,0}^{\mathrm{BS}}
=
\arccos(-\mathbf{f}_{m}^{T}(\boldsymbol{\theta}_{m})\hat{\mathbf{d}}_{0})$
and
$\varepsilon_{0}^{\mathrm{ref}}
=
\arccos(\mathbf{n}(\boldsymbol{\psi})^{T}\hat{\mathbf{d}}_{0})$.
Then, the IRS-BS channel in~\eqref{eq:HRB_entry} reduces to
$	\mathbf{H}_{\mathrm{RB}}^{\mathrm{FF}}
(\boldsymbol{\Theta},\boldsymbol{\psi})
=
c_{\mathrm{RB}}
\mathbf{u}_{\mathrm{B}}(\boldsymbol{\Theta})
\mathbf{u}_{\mathrm{R}}^{T}(\boldsymbol{\psi})$,
where $c_{\mathrm{RB}}=\lambda e^{-jk_c d_{\mathrm{RB}}}/(4\pi d_{\mathrm{RB}})$, and
\begin{align}
	[\mathbf{u}_{\mathrm{B}}(\boldsymbol{\Theta})]_{m}
	=
	\sqrt{
		G_{\mathrm{e}}(\varepsilon_{m,0}^{\mathrm{BS}})
	}
	e^{-jk_c\hat{\mathbf{d}}_{0}^{T}(\mathbf{b}_{m}-\mathbf{b}_{0})},
	\label{eq:ff_uB_def}
	\\
	[\mathbf{u}_{\mathrm{R}}(\boldsymbol{\psi})]_{n}
	=
	\sqrt{
		G_{\mathrm{e}}(\varepsilon_{0}^{\mathrm{ref}})
	}
	e^{jk_c\hat{\mathbf{d}}_{0}^{T}(\mathbf{r}_{n}(\boldsymbol{\psi})-\mathbf{r}_{0})}.
	\label{eq:ff_uR_def}
\end{align}
Following~\eqref{eq:su_an_def}, the far-field cascaded matrix becomes
\begin{equation}
	\label{eq:ff_A_rank_one}
	\mathbf{A}_{\mathrm{FF}}
	(\boldsymbol{\Theta},\boldsymbol{\psi})
	=
	c_{\mathrm{RB}}
	\mathbf{u}_{\mathrm{B}}(\boldsymbol{\Theta})
	\mathbf{s}^{T}_{\mathrm{R}}(\boldsymbol{\psi}),
\end{equation}
where
\begin{equation}
	\mathbf{s}_{\mathrm{R}}(\boldsymbol{\psi})
	=
	\mathbf{u}_{\mathrm{R}}(\boldsymbol{\psi})
	\odot
	\mathbf{h}_{\mathrm{R}}(\boldsymbol{\psi}).
\end{equation}	
For fixed $\boldsymbol{\Theta}$ and $\boldsymbol{\psi}$, the optimized far-field reflected-channel power is defined as
\begin{equation}
	\label{eq:J_FF_definition}
	J_{\mathrm{FF}}
	(\boldsymbol{\Theta},\boldsymbol{\psi})
	\triangleq
	\max_{|v_n|=1}
	\left\|
	\mathbf{A}_{\mathrm{FF}}
	(\boldsymbol{\Theta},\boldsymbol{\psi})
	\mathbf{v}
	\right\|^{2}.
\end{equation}
The fixed-orientation baseline is denoted by 
$(\boldsymbol{\Theta}_{0},\boldsymbol{\psi}_{0})$.
For $\mathcal{X}\in\{\mathrm{FF},\mathrm{NF}\}$, define
$J_{\mathcal{X}}(\boldsymbol{\Theta},\boldsymbol{\psi})$ as the optimized 
reflected-channel power under the corresponding IRS-BS channel model.
Specifically, $J_{\mathrm{FF}}$ is given by~\eqref{eq:J_FF_definition}, whereas 
$J_{\mathrm{NF}}$ is computed using the original near-field channel
$\mathbf{H}_{\mathrm{RB}}(\boldsymbol{\Theta},\boldsymbol{\psi})$.
The BS-only, IRS-only, and dual-rotation gains are defined as
\begin{equation}
	\label{eq:eta_all}
	\!\!\!\!\!\! \eta_{\mathrm{B}}^{\mathcal{X}}
	\!	 \triangleq\! 
	\frac{
		J_{\mathcal{X}}(\boldsymbol{\Theta},\boldsymbol{\psi}_{0})
	}{
		J_{\mathcal{X}}(\boldsymbol{\Theta}_{0},\boldsymbol{\psi}_{0})
	},
	\eta_{\mathrm{I}}^{\mathcal{X}}
	\! \triangleq\! 
	\frac{
		J_{\mathcal{X}}(\boldsymbol{\Theta}_{0},\boldsymbol{\psi})
	}{
		J_{\mathcal{X}}(\boldsymbol{\Theta}_{0},\boldsymbol{\psi}_{0})
	},
	\eta_{\mathrm{dual}}^{\mathcal{X}}
	\!\!\triangleq\!\!
	\frac{
		J_{\mathcal{X}}(\boldsymbol{\Theta},\boldsymbol{\psi})
	}{
		J_{\mathcal{X}}(\boldsymbol{\Theta}_{0},\boldsymbol{\psi}_{0})
	}.
\end{equation}

\begin{proposition}
	\label{prop:ff_dual_gain}
	Under the far-field approximation, the optimized reflected-channel power is
	\begin{equation}
		\label{eq:J_FF_canonical}
		J_{\mathrm{FF}}(\boldsymbol{\Theta},\boldsymbol{\psi})
		=
		N^{2}
		\bar{p}_{\mathrm{FF}}(\boldsymbol{\Theta},\boldsymbol{\psi})
		\beta_{\mathrm{FF}}(\boldsymbol{\Theta},\boldsymbol{\psi}),
	\end{equation}
	where
	\begin{align}
		\bar{p}_{\mathrm{FF}}(\boldsymbol{\Theta},\boldsymbol{\psi})
		&=
		\frac{|c_{\mathrm{RB}}|^{2}}{N}
		\left\|
		\mathbf{u}_{\mathrm{B}}(\boldsymbol{\Theta})
		\right\|^{2}
		\sum\nolimits_{n=1}^{N}
		\left|
		[\mathbf{s}_{\mathrm{R}}(\boldsymbol{\psi})]_{n}
		\right|^{2},
		\label{eq:p_bar_FF_def}
		\\
		\beta_{\mathrm{FF}}(\boldsymbol{\Theta},\boldsymbol{\psi})
		&=
		\frac{
			\left(
			\sum_{n=1}^{N}
			\left|
			[\mathbf{s}_{\mathrm{R}}(\boldsymbol{\psi})]_{n}
			\right|
			\right)^{2}
		}{
			N
			\sum_{n=1}^{N}
			\left|
			[\mathbf{s}_{\mathrm{R}}(\boldsymbol{\psi})]_{n}
			\right|^{2}
		}.
		\label{eq:beta_FF_def}
	\end{align}
	The optimized reflected-channel power is multiplicatively separable with respect to $\boldsymbol{\Theta}$ and $\boldsymbol{\psi}$, and the far-field dual-rotation gain satisfies
	\begin{equation}
		\label{eq:eta_dual_FF_separable}
		\eta_{\mathrm{dual}}^{\mathrm{FF}}
		=
		\eta_{\mathrm{B}}^{\mathrm{FF}}
		\eta_{\mathrm{I}}^{\mathrm{FF}}.
	\end{equation}
	When the user-IRS channel is further approximated by LoS component, $\beta_{\mathrm{FF}}\equiv 1$.
\end{proposition}
\begin{proof}
	From~\eqref{eq:ff_A_rank_one}, we have
	\begin{equation}
		\label{eq:ff_Av_norm}
		\|\mathbf{A}_{\mathrm{FF}}\mathbf{v}\|^{2}
		=
		|c_{\mathrm{RB}}|^{2}
		\|\mathbf{u}_{\mathrm{B}}(\boldsymbol{\Theta})\|^{2}
		\left|
		\sum_{n=1}^{N}
		[\mathbf{s}_{\mathrm{R}}(\boldsymbol{\psi})]_{n}v_n
		\right|^{2}.
	\end{equation}
	Under the unit-modulus constraint $|v_n|=1$, the triangle inequality gives
	$|\sum_{n=1}^{N}[\mathbf{s}_{\mathrm{R}}]_{n}v_n|
	\leq\sum_{n=1}^{N}|[\mathbf{s}_{\mathrm{R}}]_{n}|$,
	with equality attained by the phase-aligned choice
	$v_n^{\star}=e^{-j\angle[\mathbf{s}_{\mathrm{R}}]_{n}}$.
	Substituting $v_n^{\star}$ into~\eqref{eq:ff_Av_norm} yields the closed-form expression
	\begin{equation}
		\label{eq:J_FF_closed_proof}
		J_{\mathrm{FF}}(\boldsymbol{\Theta},\boldsymbol{\psi})
		=
		|c_{\mathrm{RB}}|^{2}
		\|\mathbf{u}_{\mathrm{B}}(\boldsymbol{\Theta})\|^{2}
		\left(
		\sum_{n=1}^{N}
		|[\mathbf{s}_{\mathrm{R}}(\boldsymbol{\psi})]_{n}|
		\right)^{2}.
	\end{equation}
	For notational simplicity, we omit the arguments of 
	$\mathbf{u}_{\mathrm{B}}(\boldsymbol{\Theta})$ and 
	$\mathbf{s}_{\mathrm{R}}(\boldsymbol{\psi})$ when no confusion occurs.
	For $\mathbf{s}_{\mathrm{R}}\neq\mathbf{0}$, \eqref{eq:J_FF_closed_proof} can be rewritten as
	\begin{equation}
		\label{eq:J_FF_recast}
		\begin{aligned}
			\!\!\!\!	J_{\mathrm{FF}}
			&\!\!=
			|c_{\mathrm{RB}}|^{2}\|\mathbf{u}_{\mathrm{B}}\|^{2}
			\cdot
			N\!\sum_{n=1}^{N}|[\mathbf{s}_{\mathrm{R}}]_{n}|^{2}
			\cdot
			\underbrace{
				\frac{\left(\sum_{n=1}^{N}|[\mathbf{s}_{\mathrm{R}}]_{n}|\right)^{2}}
				{N\sum_{n=1}^{N}|[\mathbf{s}_{\mathrm{R}}]_{n}|^{2}}}_{=\,\beta_{\mathrm{FF}}(\boldsymbol{\Theta},\boldsymbol{\psi})}\\
			&\!\!=
			N^{2}\cdot
			\underbrace{\frac{|c_{\mathrm{RB}}|^{2}}{N}
				\|\mathbf{u}_{\mathrm{B}}\|^{2}\!\sum_{n=1}^{N}|[\mathbf{s}_{\mathrm{R}}]_{n}|^{2}}_{=\,\bar{p}_{\mathrm{FF}}(\boldsymbol{\Theta},\boldsymbol{\psi})}
			\cdot\beta_{\mathrm{FF}}(\boldsymbol{\Theta},\boldsymbol{\psi}),
		\end{aligned}
	\end{equation}
	This is exactly the decomposition in~\eqref{eq:J_FF_canonical}.
	Next, define
	\begin{equation}
		\!\!\!\!	\alpha_{\mathrm{B}}(\boldsymbol{\Theta})
		\!\!=\!\!
		|c_{\mathrm{RB}}|^{2}
		\|\mathbf{u}_{\mathrm{B}}(\boldsymbol{\Theta})\|^{2}, 
		\alpha_{\mathrm{I}}(\boldsymbol{\psi})
		\!\!=\!\!
		\left(
		\sum_{n=1}^{N}
		|[\mathbf{s}_{\mathrm{R}}(\boldsymbol{\psi})]_{n}|
		\right)^{2}.
	\end{equation}
	Then~\eqref{eq:J_FF_closed_proof} can be rewritten as
	\begin{equation}
		J_{\mathrm{FF}}(\boldsymbol{\Theta},\boldsymbol{\psi})
		=
		\alpha_{\mathrm{B}}(\boldsymbol{\Theta})
		\alpha_{\mathrm{I}}(\boldsymbol{\psi}),
	\end{equation}
	which shows the multiplicative separability between BS boresight steering and IRS orientation.
	Therefore, the BS-only and IRS-only rotation gain can be respectively expressed as
	\begin{align}
		\!\eta_{\mathrm{B}}^{\mathrm{FF}}
		&\!=\!
		\frac{
			J_{\mathrm{FF}}(\boldsymbol{\Theta},\boldsymbol{\psi}_{0})
		}{
			J_{\mathrm{FF}}(\boldsymbol{\Theta}_{0},\boldsymbol{\psi}_{0})
		}
		\!=\!
		\frac{
			{\alpha}_{\mathrm{B}}(\boldsymbol{\Theta})
		}{
			{\alpha}_{\mathrm{B}}(\boldsymbol{\Theta}_{0})
		},
		\label{eq:eta_B_FF_proof}
	\end{align}
	\begin{align}
		\!\eta_{\mathrm{I}}^{\mathrm{FF}}
		\!=\!
		\frac{
			J_{\mathrm{FF}}(\boldsymbol{\Theta}_{0},\boldsymbol{\psi})
		}{
			J_{\mathrm{FF}}(\boldsymbol{\Theta}_{0},\boldsymbol{\psi}_{0})
		}
		\!=\!
		\frac{
			{\alpha}_{\mathrm{I}}(\boldsymbol{\psi})
		}{
			{\alpha}_{\mathrm{I}}(\boldsymbol{\psi}_{0})
		}.
		\label{eq:eta_I_FF_proof}
	\end{align}
	For dual rotation, we have
	\begin{equation}
		\begin{aligned}
			\!\!	\eta_{\mathrm{dual}}^{\mathrm{FF}}
			=
			\frac{
				J_{\mathrm{FF}}(\boldsymbol{\Theta},\boldsymbol{\psi})
			}{
				J_{\mathrm{FF}}(\boldsymbol{\Theta}_{0},\boldsymbol{\psi}_{0})
			}=
			\frac{
				\alpha_{\mathrm{B}}(\boldsymbol{\Theta})
			}{
				\alpha_{\mathrm{B}}(\boldsymbol{\Theta}_{0})
			}
			\frac{
				\alpha_{\mathrm{I}}(\boldsymbol{\psi})
			}{
				\alpha_{\mathrm{I}}(\boldsymbol{\psi}_{0})
			}
			=
			\eta_{\mathrm{B}}^{\mathrm{FF}}
			\eta_{\mathrm{I}}^{\mathrm{FF}}.
		\end{aligned}
	\end{equation}
	
	For the LoS-dominant user-IRS channel, the far-field approximation gives identical magnitudes for all entries of $\mathbf{h}_{\mathrm R}(\boldsymbol{\psi})$ and $\mathbf{u}_{\mathrm R}(\boldsymbol{\psi})$.
	Hence, all entries of $\mathbf{s}_{\mathrm R}(\boldsymbol{\psi})$ have the same magnitude.
	Let $|[\mathbf{s}_{\mathrm R}(\boldsymbol{\psi})]_n|=s_0$ for all $n$.
	Then, from~\eqref{eq:beta_FF_def}, we have
	\begin{equation}
		\beta_{\mathrm{FF}}
		=
		\frac{(Ns_{0})^{2}}{N\cdot Ns_{0}^{2}}
		=
		1.
	\end{equation}
	This completes the proof.
\end{proof}
It is observed from Proposition~\ref{prop:ff_dual_gain} that, under the far-field, the optimized reflected channel power can be factorized into $\alpha_{\mathrm{B}}(\boldsymbol{\Theta})$ and  $\alpha_{\mathrm{I}}(\boldsymbol{\psi})$. 
Therefore, the far-field dual-rotation gain is exactly the product of the BS-only and IRS-only rotation gains, i.e.,
$\eta_{\mathrm{dual}}^{\mathrm{FF}}=\eta_{\mathrm B}^{\mathrm{FF}}\eta_{\mathrm I}^{\mathrm{FF}}$.
This implies that joint BS-boresight and IRS-orientation optimization leads to a multiplicative combination of the two independent single-rotation gains rather than a distinct coupling gain.
Moreover, when the user-IRS link is LoS-dominant, $\beta_{\mathrm{FF}}=1$, indicating that the IRS phase shifts can coherently combine the reflected components without coherent-combining loss.

The above separability, however, generally fails to hold in the near field. 
The reason is that the IRS-BS channel cannot be approximated by a common propagation direction across the IRS and BS arrays. 
Instead, the propagation distances and directions vary with the IRS elements and BS antennas, which couples the IRS orientation with the BS boresight steering.
To characterize this coupling, the optimized near-field reflected-channel power is analyzed.
From~\eqref{eq:HRB_entry} and~\eqref{eq:su_an_def}, we have
\begin{equation}
	[\mathbf{A}_{\mathrm{NF}}]_{m,n}
	=
	[\mathbf{H}_{\mathrm{RB}}(\boldsymbol{\Theta},\boldsymbol{\psi})]_{m,n}
	[\mathbf{h}_{\mathrm R}(\boldsymbol{\psi})]_n, 
\end{equation}
and
$
\mathbf a_n(\boldsymbol{\Theta},\boldsymbol{\psi})
=
[\mathbf A_{\mathrm{NF}}(\boldsymbol{\Theta},\boldsymbol{\psi})]_{:,n}
$.
Here, $\mathbf a_n(\boldsymbol{\Theta},\boldsymbol{\psi})$ depends on the distance $d_{n,m}(\boldsymbol{\psi})$, the IRS-side reflection angle $\varepsilon_{n,m}^{\mathrm{ref}}(\boldsymbol{\psi})$, and the BS-side receive angle $\varepsilon_{m,n}^{\mathrm{BS}}(\boldsymbol{\theta}_{m},\boldsymbol{\psi})$, which jointly involve IRS orientation and BS boresight steering. 
For fixed $\boldsymbol{\Theta}$ and $\boldsymbol{\psi}$, the optimized near-field reflected-channel power is defined as
\begin{equation}
	J_{\mathrm{NF}}(\boldsymbol{\Theta},\boldsymbol{\psi})
	\triangleq
	\max_{|v_n|=1}
	\left\|\mathbf{A}_{\mathrm{NF}}(\boldsymbol{\Theta},\boldsymbol{\psi})\mathbf{v}\right\|^{2},
\end{equation}
and the rotation gains $\eta_{\mathrm{B}}^{\mathrm{NF}}$, $\eta_{\mathrm{I}}^{\mathrm{NF}}$, and $\eta_{\mathrm{dual}}^{\mathrm{NF}}$ are obtained from~\eqref{eq:eta_all} by setting $\mathcal{X}=\mathrm{NF}$.

\begin{proposition} 
	\label{prop:nf_dual_coupling}
	Under the LoS-dominant user-IRS channel, define
	\begin{equation}
		\label{eq:p_bar_NF_def}
		\bar{p}_{\mathrm{NF}}(\boldsymbol{\Theta},\boldsymbol{\psi})
		\triangleq
		\frac{\kappa_{0}(\boldsymbol{\psi})}{N}
		\!\!\sum_{m=1}^{M}
		\!\sum_{n=1}^{N}\!\!
		\frac{
			G_{\mathrm{e}}(\varepsilon_{n,m}^{\mathrm{ref}}(\boldsymbol{\psi}))
			G_{\mathrm{e}}(\varepsilon_{m,n}^{\mathrm{BS}}(\boldsymbol{\theta}_{m},\boldsymbol{\psi}))
		}{
			d_{n,m}^{2}(\boldsymbol{\psi})
		}.
	\end{equation}
	Then the optimized near-field reflected-channel power satisfies
	\begin{equation}
		\label{eq:J_NF_power_bound}
		N\bar{p}_{\mathrm{NF}}(\boldsymbol{\Theta},\boldsymbol{\psi})
		\le
		J_{\mathrm{NF}}(\boldsymbol{\Theta},\boldsymbol{\psi})
		\le
		N^2\bar{p}_{\mathrm{NF}}(\boldsymbol{\Theta},\boldsymbol{\psi}) .
	\end{equation}
	Equivalently, the normalized near-field combining efficiency
	\begin{equation}
		\label{eq:beta_NF_def}
		\beta_{\mathrm{NF}}(\boldsymbol{\Theta},\boldsymbol{\psi})
		\triangleq
		\frac{
			J_{\mathrm{NF}}(\boldsymbol{\Theta},\boldsymbol{\psi})
		}{
			N^2\bar{p}_{\mathrm{NF}}(\boldsymbol{\Theta},\boldsymbol{\psi})
		}
	\end{equation}
	satisfies
	\begin{equation}
		\label{eq:beta_NF_range}
		\frac{1}{N}
		\le
		\beta_{\mathrm{NF}}(\boldsymbol{\Theta},\boldsymbol{\psi})
		\le
		1 .
	\end{equation}
	Consequently, the near-field dual-rotation gain can be written as
	\begin{equation}
		\label{eq:eta_dual_NF_Gp_Gbeta}
		\eta_{\mathrm{dual}}^{\mathrm{NF}}
		=
		G_{p}^{\mathrm{NF}}
		G_{\beta}^{\mathrm{NF}},
	\end{equation}
	where $
	G_{p}^{\mathrm{NF}}
	\triangleq
	\bar{p}_{\mathrm{NF}}(\boldsymbol{\Theta},\boldsymbol{\psi})
	/
	\bar{p}_{\mathrm{NF}}(\boldsymbol{\Theta}_{0},\boldsymbol{\psi}_{0})
	$,
	and
	$	G_{\beta}^{\mathrm{NF}}
	\triangleq
	\beta_{\mathrm{NF}}(\boldsymbol{\Theta},\boldsymbol{\psi})
	/
	\beta_{\mathrm{NF}}(\boldsymbol{\Theta}_{0},\boldsymbol{\psi}_{0})
	.
	$
\end{proposition}

\begin{proof}
	For compact notation, the arguments
	$(\boldsymbol{\Theta},\boldsymbol{\psi})$ are omitted when no confusion occurs.
	Under the LoS-dominant user-IRS channel, we have
		$	\mathbf h_{\mathrm R}(\boldsymbol{\psi})
		=
		c_{\mathrm R,1}
		\sqrt{G_{1}^{\mathrm{inc}}(\boldsymbol{\psi})}
		\mathbf a_{\mathrm R,1}(\boldsymbol{\psi})$.
	Since each entry of $\mathbf a_{\mathrm R,1}(\boldsymbol{\psi})$ has unit modulus,
	\begin{equation}
		\left|
		[\mathbf h_{\mathrm R}(\boldsymbol{\psi})]_n
		\right|^2
		=
		|c_{\mathrm R,1}|^2
		G_{1}^{\mathrm{inc}}(\boldsymbol{\psi}) ,
		\quad n=1,\ldots,N .
	\end{equation}
	Using
	$[\mathbf A_{\mathrm{NF}}]_{m,n}
	=
	[\mathbf H_{\mathrm{RB}}]_{m,n}
	[\mathbf h_{\mathrm R}]_n$,
	we obtain
	\begin{equation}
		\begin{aligned}
			\label{eq:NF_column_power_expanded}
			\!\!\! \sum_{n=1}^{N}\!\!
			\left\|
			\mathbf a_n
			\right\|^{2}\!
			&\!\!=\!\!
			\kappa_{0}(\boldsymbol{\psi})
			\!\!\sum_{m=1}^{M}
			\!\sum_{n=1}^{N}\!\!
			\frac{
				G_{\mathrm{e}}(\varepsilon_{n,m}^{\mathrm{ref}}(\boldsymbol{\psi}))
				G_{\mathrm{e}}(\varepsilon_{m,n}^{\mathrm{BS}}(\boldsymbol{\theta}_{m},\boldsymbol{\psi}))
			}{
				d_{n,m}^{2}(\boldsymbol{\psi})
			}\\
			&	=
			N
			\bar{p}_{\mathrm{NF}}(\boldsymbol{\Theta},\boldsymbol{\psi}).
		\end{aligned}
	\end{equation}
	For any feasible $\mathbf v$, we have
		$	\mathbf A_{\mathrm{NF}}\mathbf v
		=
		\sum_{n=1}^{N}
		v_n\mathbf a_n$.
	The triangle inequality gives
		$	\left\|
		\mathbf A_{\mathrm{NF}}\mathbf v
		\right\|
		\le
		\sum_{n=1}^{N}
		\left\|
		\mathbf a_n
		\right\|$.
	Maximizing over all feasible $\mathbf v$ yields
	\begin{equation}
		J_{\mathrm{NF}}(\boldsymbol{\Theta},\boldsymbol{\psi})
		\le
		\left(
		\sum_{n=1}^{N}
		\left\|
		\mathbf a_n
		\right\|
		\right)^2 .
	\end{equation}
	By the Cauchy--Schwarz inequality,
	\begin{equation}
		\left(
		\sum_{n=1}^{N}
		\left\|
		\mathbf a_n
		\right\|
		\right)^2
		\le
		N
		\sum_{n=1}^{N}
		\left\|
		\mathbf a_n
		\right\|^2
		=
		N^2
		\bar{p}_{\mathrm{NF}}(\boldsymbol{\Theta},\boldsymbol{\psi}) .
	\end{equation}
	To obtain a lower bound, choose
	$v_n=e^{j\varphi_n}$, where $\varphi_n$ are independent and uniformly distributed over $[0,2\pi)$.
	Then $\mathbb E[v_i^{*}v_n]=0$ for $i\neq n$, and
	\begin{equation}
		\begin{aligned}
			\mathbb E
			\left[
			\left\|
			\sum_{n=1}^{N}
			v_n\mathbf a_n
			\right\|^2
			\right]
			&=
			\sum_{n=1}^{N}
			\left\|
			\mathbf a_n
			\right\|^2
			=
			N
			\bar{p}_{\mathrm{NF}}(\boldsymbol{\Theta},\boldsymbol{\psi}) .
		\end{aligned}
	\end{equation}
	Therefore, there exists a feasible phase vector satisfying
	\begin{equation}
		\left\|
		\sum_{n=1}^{N}
		v_n\mathbf a_n
		\right\|^2
		\ge
		N
		\bar{p}_{\mathrm{NF}}(\boldsymbol{\Theta},\boldsymbol{\psi}) .
	\end{equation}
	Since $J_{\mathrm{NF}}(\boldsymbol{\Theta},\boldsymbol{\psi})$ is the maximum over all feasible phase vectors, we have
		$	J_{\mathrm{NF}}(\boldsymbol{\Theta},\boldsymbol{\psi})
		\ge
		N
		\bar{p}_{\mathrm{NF}}(\boldsymbol{\Theta},\boldsymbol{\psi})$.
	Combining the lower and upper bounds proves~\eqref{eq:J_NF_power_bound}.
	Dividing~\eqref{eq:J_NF_power_bound} by
	$N^2\bar{p}_{\mathrm{NF}}(\boldsymbol{\Theta},\boldsymbol{\psi})$
	gives~\eqref{eq:beta_NF_range}.
	Furthermore, from~\eqref{eq:beta_NF_def}, we have
	\begin{equation}
		J_{\mathrm{NF}}(\boldsymbol{\Theta},\boldsymbol{\psi})
		=
		N^{2}
		\bar{p}_{\mathrm{NF}}(\boldsymbol{\Theta},\boldsymbol{\psi})
		\beta_{\mathrm{NF}}(\boldsymbol{\Theta},\boldsymbol{\psi}) .
	\end{equation}
	Taking the ratio between
	$(\boldsymbol{\Theta},\boldsymbol{\psi})$ and
	$(\boldsymbol{\Theta}_{0},\boldsymbol{\psi}_{0})$ yields
	\begin{equation}
		\begin{aligned}
			\eta_{\mathrm{dual}}^{\mathrm{NF}}
			=
			\frac{
				J_{\mathrm{NF}}(\boldsymbol{\Theta},\boldsymbol{\psi})
			}{
				J_{\mathrm{NF}}(\boldsymbol{\Theta}_{0},\boldsymbol{\psi}_{0})
			}
			=
			\frac{
				\bar{p}_{\mathrm{NF}}(\boldsymbol{\Theta},\boldsymbol{\psi})
			}{
				\bar{p}_{\mathrm{NF}}(\boldsymbol{\Theta}_{0},\boldsymbol{\psi}_{0})
			}
			\frac{
				\beta_{\mathrm{NF}}(\boldsymbol{\Theta},\boldsymbol{\psi})
			}{
				\beta_{\mathrm{NF}}(\boldsymbol{\Theta}_{0},\boldsymbol{\psi}_{0})
			} .
		\end{aligned}
	\end{equation}
	This completes the proof.
\end{proof}

Proposition~\ref{prop:nf_dual_coupling} shows that the near-field dual-rotation gain is decomposed into the average column-power gain $G_{p}^{\mathrm{NF}}$ of the near-field reflected channel matrix and the coherent-combining gain $G_{\beta}^{\mathrm{NF}}$. 
Since both of these factors depend on $\boldsymbol{\Theta}$ and $\boldsymbol{\psi}$, we cannot derive the near-field dual rotation gain from the rotation gains of the BS and the IRS alone, which motivates the coordinated design of BS boresight steering and IRS panel orientation.
Moreover, the bound $\beta_{\mathrm{NF}}\in[1/N,1]$ indicates that the IRS phase shifts may not always achieve fully coherent combining in the near-field regime. 
When the IRS-BS propagation directions are nearly identical, $\beta_{\mathrm{NF}}$ approaches one and the conventional $N^{2}$ power scaling can be recovered. 
In contrast, when the near-field direction variation becomes pronounced, $\beta_{\mathrm{NF}}$ may decrease and the reflected-channel power may degrade toward linear scaling with $N$.
To characterize the geometry-dependent near-field direction variation, we analyze the impact of the IRS aperture-to-distance ratio on the near-field rotation behavior.
Define
$\rho_{n,m}(\boldsymbol{\psi})=\mathbf{n}(\boldsymbol{\psi})^{T}\hat{\mathbf{d}}_{n,m}(\boldsymbol{\psi})$
as the alignment between the IRS normal and the IRS element-to-BS antenna direction.
The alignment variation is
\begin{equation}
	\label{eq:Delta_def}
	\!\!\!\!	\Delta(\boldsymbol{\psi})
	\!\!	=\!\!
	\left[
	\frac{1}{MN}
	\sum\nolimits_{m=1}^{M}
	\sum\nolimits_{n=1}^{N}
	\left(
	\rho_{n,m}(\boldsymbol{\psi})
	-
	\bar{\rho}(\boldsymbol{\psi})
	\right)^{2}
	\right]^{1/2},
\end{equation}
where
$
\bar{\rho}(\boldsymbol{\psi})=(MN)^{-1}\sum_{m=1}^{M}\sum_{n=1}^{N}\rho_{n,m}(\boldsymbol{\psi}).
$
Let $\xi=D_{\mathrm{R}}/d_{\mathrm{RB}}$ denote the normalized IRS aperture-to-distance ratio.
The aggregate IRS-side reflection gain is
$		G(\boldsymbol{\psi})
=
G_{0,\mathrm{IRS}}
\sum_{m=1}^{M}
\sum_{n=1}^{N}
\rho_{n,m}^{2p_{\mathrm{IRS}}}(\boldsymbol{\psi})$.
For numerical validation, the IRS-side reflection-gain variance is defined as
\begin{equation}
\mathcal{S}(\boldsymbol{\psi})
=
(MN)^{-1}
\sum_{m=1}^{M}
\sum_{n=1}^{N}
(
G_{n,m}^{\mathrm{ref}}(\boldsymbol{\psi})
-
\bar{G}^{\mathrm{ref}}(\boldsymbol{\psi})
)^{2},
\end{equation}
where
$\bar{G}^{\mathrm{ref}}(\boldsymbol{\psi})
=
(MN)^{-1}\sum_{m=1}^{M}\sum_{n=1}^{N}
G_{n,m}^{\mathrm{ref}}(\boldsymbol{\psi})$.

\begin{proposition}
\label{prop:xi_tradeoff}
Assume $0<\rho_{n,m}(\boldsymbol{\psi})\leq 1$ and $p_{\mathrm{IRS}}\geq1/2$.
Under the compact-BS condition $D_{\mathrm{B}}\ll D_{\mathrm{R}}$, where $D_{\mathrm{B}}$ denotes the BS aperture diagonal, the alignment variation satisfies
\begin{equation}
	\label{eq:Delta_xi_bound}
	\Delta(\boldsymbol{\psi})
	\leq
	C_{\Delta}(\boldsymbol{\psi})\xi
	+
	O(\xi^{2}),
\end{equation}
where $C_{\Delta}(\boldsymbol{\psi})$ is independent of $\xi$ for fixed geometry.
Moreover, the aggregate IRS-side reflection gain satisfies
\begin{equation}
	\label{eq:gain_variation_upper}
	\!\!	G(\boldsymbol{\psi})
	\leq
	MNG_{0,\mathrm{IRS}}\bar{\rho}(\boldsymbol{\psi})
	\leq
	MNG_{0,\mathrm{IRS}}
	\left(
	1-\Delta^{2}(\boldsymbol{\psi})
	\right).
\end{equation}
\end{proposition}

\begin{proof}
Let $\hat{\mathbf d}_{0}=(\mathbf b_{0}-\mathbf r_{0})/d_{\mathrm{RB}}$,
$\epsilon_{\mathrm{ap}}=(D_{\mathrm B}+D_{\mathrm R})/d_{\mathrm{RB}}$, and
$
\boldsymbol{\delta}_{m,n}(\boldsymbol{\psi})
=
\bar{\mathbf b}_{m}
-
\mathbf R(\boldsymbol{\psi})\bar{\mathbf r}_{n}.
$
A first-order expansion of the unit direction gives
\begin{equation}
	\label{eq:dnm_taylor}
	\hat{\mathbf d}_{n,m}(\boldsymbol{\psi})
	=
	\hat{\mathbf d}_{0}
	+
	\frac{
		(\mathbf I-\hat{\mathbf d}_{0}\hat{\mathbf d}_{0}^{T})
		\boldsymbol{\delta}_{m,n}(\boldsymbol{\psi})
	}{
		d_{\mathrm{RB}}
	}
	+
	O(\epsilon_{\mathrm{ap}}^{2}).
\end{equation}
Define
$\mathbf q(\boldsymbol{\psi})
=
(\mathbf I-\hat{\mathbf d}_{0}\hat{\mathbf d}_{0}^{T})
\mathbf n(\boldsymbol{\psi})$.
Using
$\rho_{n,m}(\boldsymbol{\psi})
=
\mathbf n^{T}(\boldsymbol{\psi})
\hat{\mathbf d}_{n,m}(\boldsymbol{\psi})$
and averaging over all $(n,m)$, we obtain
\begin{equation}
	\label{eq:rho_centered_expansion}
	\!\!	\rho_{n,m}(\boldsymbol{\psi})
	\!\!	-	\!\!
	\bar{\rho}(\boldsymbol{\psi})
	\!\!	=	\!\!
	\frac{
		\mathbf q^{T}(\boldsymbol{\psi})
		\big(
		\boldsymbol{\delta}_{m,n}(\boldsymbol{\psi})
		-
		\bar{\boldsymbol{\delta}}(\boldsymbol{\psi})
		\big)
	}{
		d_{\mathrm{RB}}
	}
	+
	O(\epsilon_{\mathrm{ap}}^{2}),
\end{equation}
where
$\bar{\boldsymbol{\delta}}(\boldsymbol{\psi})
=
(MN)^{-1}
\sum_{m=1}^{M}
\sum_{n=1}^{N}
\boldsymbol{\delta}_{m,n}(\boldsymbol{\psi})$.
Since $\|\mathbf q(\boldsymbol{\psi})\|\leq1$,
$\|\bar{\mathbf b}_{m}\|\leq D_{\mathrm B}/2$, and
$\|\mathbf R(\boldsymbol{\psi})\bar{\mathbf r}_{n}\|\leq D_{\mathrm R}/2$,
\eqref{eq:rho_centered_expansion} and~\eqref{eq:Delta_def} yield
\begin{equation}
	\label{eq:Delta_pre_bound}
	\Delta(\boldsymbol{\psi})
	\leq
	C_{1}(\boldsymbol{\psi})
	\frac{D_{\mathrm R}}{d_{\mathrm{RB}}}
	+
	C_{2}(\boldsymbol{\psi})
	\frac{D_{\mathrm B}}{d_{\mathrm{RB}}}
	+
	O(\epsilon_{\mathrm{ap}}^{2}),
\end{equation}
where $C_{1}(\boldsymbol{\psi})$ and $C_{2}(\boldsymbol{\psi})$ are bounded geometry-dependent constants.
Using $\xi=D_{\mathrm R}/d_{\mathrm{RB}}$, the first-order terms in~\eqref{eq:Delta_pre_bound} become
$C_{1}(\boldsymbol{\psi})\xi
+
C_{2}(\boldsymbol{\psi})(D_{\mathrm B}/D_{\mathrm R})\xi$.
Under $D_{\mathrm B}\ll D_{\mathrm R}$,
$\epsilon_{\mathrm{ap}}=(1+D_{\mathrm B}/D_{\mathrm R})\xi=O(\xi)$.
Hence, 
\begin{equation}
	\Delta(\boldsymbol{\psi})
	\leq
	C_{\Delta}(\boldsymbol{\psi})\xi
	+
	O(\xi^{2}),
\end{equation}
where $
C_{\Delta}(\boldsymbol{\psi})
=
C_{1}(\boldsymbol{\psi})
+
C_{2}(\boldsymbol{\psi})
\frac{D_{\mathrm B}}{D_{\mathrm R}}.
$
This proves~\eqref{eq:Delta_xi_bound}.

Since $0<\rho_{n,m}(\boldsymbol{\psi})\leq1$ and
$p_{\mathrm{IRS}}\geq1/2$, we have
$\rho_{n,m}^{2p_{\mathrm{IRS}}}(\boldsymbol{\psi})
\leq
\rho_{n,m}(\boldsymbol{\psi})$.
Thus,
\begin{equation}
	\label{eq:G_rho_bound}
	\!\!	G(\boldsymbol{\psi})
	\!	=\!
	G_{0,\mathrm{IRS}}
	\sum_{m=1}^{M}
	\sum_{n=1}^{N}
	\rho_{n,m}^{2p_{\mathrm{IRS}}}(\boldsymbol{\psi})
	\!\leq\!
	MNG_{0,\mathrm{IRS}}
	\bar{\rho}(\boldsymbol{\psi}).
\end{equation}
Expanding the squared term in~\eqref{eq:Delta_def} gives
\begin{equation}
	\Delta^{2}(\boldsymbol{\psi})
	=
	\overline{\rho^{2}}(\boldsymbol{\psi})
	-
	\bar{\rho}^{2}(\boldsymbol{\psi}),
\end{equation}
where
$
\overline{\rho^{2}}(\boldsymbol{\psi})
=
(MN)^{-1}
\sum_{m=1}^{M}
\sum_{n=1}^{N}
\rho_{n,m}^{2}(\boldsymbol{\psi})$.
Since $0<\rho_{n,m}(\boldsymbol{\psi})\leq1$ for all $(n,m)$,
we have
$\rho_{n,m}^{2}(\boldsymbol{\psi})
\leq
\rho_{n,m}(\boldsymbol{\psi})$,
which implies
$\overline{\rho^{2}}(\boldsymbol{\psi})
\leq
\bar{\rho}(\boldsymbol{\psi})$.
It follows that
$
\Delta^{2}(\boldsymbol{\psi})
\leq
\bar{\rho}(\boldsymbol{\psi})
\left(
1-\bar{\rho}(\boldsymbol{\psi})
\right)
\leq
1-\bar{\rho}(\boldsymbol{\psi}).
$
Therefore,
$\bar{\rho}(\boldsymbol{\psi})
\leq
1-\Delta^{2}(\boldsymbol{\psi})$.
Combining this inequality with~\eqref{eq:G_rho_bound} proves~\eqref{eq:gain_variation_upper}.
This completes the proof.		
\end{proof}

This result indicates that the aperture-to-distance ratio $\xi$ captures near-field direction variation in the IRS-BS link.
When $\xi\ll 1$, the IRS aperture is small relative to the IRS-BS distance, and
the IRS-BS propagation directions become nearly identical.
In this case, the far-field separability in Proposition~\ref{prop:ff_dual_gain}
approximately holds.
For a larger $\xi$, the bound in Proposition~\ref{prop:xi_tradeoff} allows
stronger spatial variation of the IRS-BS propagation directions across the IRS
and BS arrays.
Meanwhile, the IRS gain is upper-bounded by
$MNG_{0,\mathrm{IRS}}(1-\Delta^{2}(\boldsymbol{\psi}))$, indicating that
stronger direction variation may reduce the achievable average reflection gain.
Therefore, $\xi$ can be regarded as a useful geometric indicator for assessing
the potential strength of near-field rotation coupling.

\section{Multi-User System}
\label{sec:multi_user}

In this section, we consider the general multi-user cases.
Specifically, we solve~(P1) by developing an AO-based algorithm for sum-rate maximization.

\subsection{Receive Beamforming Update}
\label{subsec:mu_mmse}

For fixed $\boldsymbol{\Theta}$, $\boldsymbol{\psi}$, and $\mathbf{v}$, the receive beamforming subproblem is separable over users.
Since the SINR in~\eqref{eq:sinr} is invariant to any nonzero scaling of $\mathbf{w}_k$, the optimal receive combiner is given by~\cite{hu2025ma_uplink_tvt}
\begin{equation}
\label{eq:mmse_combiner}
\mathbf{w}_k^{\star}
=
\frac{
	\mathbf{C}_k^{-1}
	\mathbf{h}_k(\boldsymbol{\Theta},\boldsymbol{\psi},\mathbf{v})
}{
	\left\|
	\mathbf{C}_k^{-1}
	\mathbf{h}_k(\boldsymbol{\Theta},\boldsymbol{\psi},\mathbf{v})
	\right\|
},
\end{equation}
where
$\mathbf{C}_k
\triangleq
\sum_{j\neq k}
P_j\mathbf{h}_j\mathbf{h}_j^{H}
+
\sigma^{2}\mathbf{I}_M$
denotes the interference-plus-noise covariance matrix for user~$k$.

\subsection{IRS Phase-Shift Update}
\label{subsec:mu_phase}
For fixed $\mathbf{W}$, $\boldsymbol{\Theta}$, and $\boldsymbol{\psi}$, the phase-shift subproblem is non-convex due to unit-modulus constraints.
Unlike the single-user case, the IRS phase vector affects all users, and the coordinate-wise update in~\eqref{eq:su_phase_update} is not directly applicable.
We therefore update the IRS phase shifts by an RCG method based on the FP transformations~\cite{shen2018fractional}.

For brevity, define
$u_{k,j} \triangleq \sqrt{P_j}\mathbf{w}_k^H\mathbf{h}_j$
and
$\varrho_k \triangleq \sum_{j=1}^{K}|u_{k,j}|^2 + \sigma^2\|\mathbf{w}_k\|^2$.
Then, $|u_{k,k}|^2$, $\sum_{j\neq k}|u_{k,j}|^2$, and $\varrho_k$ denote the desired signal, interference, and total received powers of user $k$, respectively.	
At FP iteration~$i$, given $\mathbf v^{(i-1)}$, the auxiliary variables are updated as $\mu_k^{(i)}=\gamma_k^{(i-1)}$ and $\nu_k^{(i)}	= 	\sqrt{1+\mu_k^{(i)}}\,u_{k,k}^{(i-1)}/\varrho_k^{(i-1)}$. 
For given $\{\mu_k^{(i)},\nu_k^{(i)}\}_{k=1}^{K}$, rewrite the composite channel in~\eqref{eq:composite_channel} as
\begin{equation}
\label{eq:composite_channel_decomp}
\mathbf{h}_j
=
\mathbf{h}_{\mathrm{B},j}
+
\mathbf{G}_j\mathbf{v},
\end{equation}
where
$\mathbf{G}_j
\triangleq
\mathbf{H}_{\mathrm{RB}}\mathrm{diag}(\mathbf{h}_{\mathrm{R},j})
\in\mathbb{C}^{M\times N}$.
Substituting~\eqref{eq:composite_channel_decomp} into the FP-transformed objective and removing the terms independent of $\mathbf{v}$, the phase-shift subproblem at FP iteration~$i$ can be written as
\begin{subequations}
\label{eq:mu_qcqp}
\begin{align}
	(\mathrm{P3}):\ \max_{\mathbf{v}}\ &
	\tilde{R}(\mathbf{v})
	\triangleq
	\mathbf{v}^H \mathbf{Q}^{(i)} \mathbf{v} +
	2\Re\{(\mathbf{q}^{(i)})^H \mathbf{v}\}
	\label{eq:P3_obj}\\
	\text{s.t.}\ & |v_n|=1,\ \forall n\in\mathcal{N},
	\label{eq:P3_const}
\end{align}
\end{subequations}
where the Hermitian matrix $\mathbf{Q}^{(i)} \in \mathbb{C}^{N \times N}$ is given by
\begin{equation}
\label{eq:mu_Q_matrix}
\mathbf{Q}^{(i)} = -\sum\nolimits_{k=1}^{K}
\frac{|\nu_k^{(i)}|^2}{\ln 2}
\sum\nolimits_{j=1}^{K} P_j\, \mathbf{G}_j^H \mathbf{w}_k \mathbf{w}_k^H \mathbf{G}_j.
\end{equation}
Define
$\mathbf{c}_k^{(i)}
\triangleq
\sqrt{(1+\mu_k^{(i)})P_k}\mathbf{G}_k^{H}\mathbf{w}_k\in\mathbb{C}^{N}$,
$\mathbf{d}_k
\triangleq
\sum_{j=1}^{K}P_j b_{k,j}\mathbf{G}_j^{H}\mathbf{w}_k\in\mathbb{C}^{N}$,
and
$b_{k,j}\triangleq\mathbf{w}_k^{H}\mathbf{h}_{\mathrm{B},j}$.
Then, $\mathbf{q}^{(i)}\in\mathbb{C}^{N}$ is given by
\begin{equation}
\label{eq:mu_q_vector}
\mathbf{q}^{(i)}
=
\sum\nolimits_{k=1}^{K}
\frac{\nu_k^{(i)}}{\ln 2}
\left[
\mathbf{c}_k^{(i)}
-
(\nu_k^{(i)})^{*}\mathbf{d}_k
\right].
\end{equation}
Since $\mathbf{Q}^{(i)}\preceq\mathbf{0}$, $\tilde{R}(\mathbf{v})$ is concave in the ambient Euclidean space, while problem~(P3) remains non-convex due to the unit-modulus constraints.
We update $\mathbf{v}$ by RCG on the complex-circle product manifold~\cite{absil2008optimization}, using the Euclidean gradient
\begin{equation}
\label{eq:FP_Euclidean_grad}
\nabla \tilde{R}(\mathbf{v}) = 2\mathbf{Q}^{(i)}\mathbf{v} + 2\mathbf{q}^{(i)}.
\end{equation}
The Riemannian gradient $\mathrm{grad}\,\tilde{R}(\mathbf{v})$ is obtained by projecting~\eqref{eq:FP_Euclidean_grad} onto the tangent space of the complex-circle manifold.

\subsection{Rotation Update}
\label{subsec:mu_rotation}
For fixed $\mathbf W$ and $\mathbf v$, the rotation subproblem remains non-convex due to the nonlinear dependence of the composite channels on $\boldsymbol{\Theta}$ and $\boldsymbol{\psi}$.
As in the single-user case, projected-gradient methods are used for the rotation update.
For the BS-side update, the equivalent boresight vector $\mathbf f_m=\mathbf f_m(\boldsymbol{\theta}_m)$ in~\eqref{eq:boresight_vector} is used instead of the angular variable $\boldsymbol{\theta}_m$.
The two gradients are related by
\begin{equation}
\label{eq:mu_theta_f_relation}
\nabla_{\boldsymbol{\theta}_m}R_{\mathrm{sum}}
=
\left(
\frac{\partial \mathbf f_m}
{\partial \boldsymbol{\theta}_m}
\right)^{T}
\nabla_{\mathbf f_m}R_{\mathrm{sum}},
\end{equation}
where the angular Jacobian is given in~\eqref{eq:su_boresight_jacobian}.
The feasible set of $\mathbf f_m$ is determined by the RA angular constraints.
Using $\mathbf f_m$ avoids repeated applications of the angular Jacobian in the multi-user sum-rate gradient.
Smoothing is applied only for gradient evaluation at the boundary of $G_{\mathrm e}$, while the exact model in~\eqref{eq:gain_pattern_modeling} is used for objective and feasibility evaluations.

\subsubsection{Sum-Rate Gradient Evaluation}

We first derive the gradient of $R_{\mathrm{sum}}$ with respect to
$\boldsymbol{\chi}\in\{\mathbf f_m,\boldsymbol{\psi}\}$.
For user $k$, let
$I_k\triangleq\sum_{j\neq k}P_j|\mathbf{w}_k^{H}\mathbf{h}_j|^{2}$
and
$Z_k\triangleq\sigma^{2}\|\mathbf{w}_k\|^{2}$.
Define the signal and interference weights as
\begin{equation}
\label{eq:weight_def}
\zeta_k^{\mathrm{sig}}
\triangleq
\frac{P_k\mathbf{w}_k^{H}\mathbf{h}_k}{I_k+Z_k},
\quad
\zeta_{k,j}^{\mathrm{int}}
\triangleq
\frac{\gamma_k P_j(\mathbf{w}_k^{H}\mathbf{h}_j)}{I_k+Z_k},
\quad j\neq k.
\end{equation}
Let
$\mathbf{D}_{j,\boldsymbol{\chi}}
\triangleq
\partial\mathbf{h}_{j}/\partial\boldsymbol{\chi}^{T}$
denote the channel-derivative matrix of $\mathbf{h}_{j}$ with respect to
$\boldsymbol{\chi}$.
By differentiating
$R_{\mathrm{sum}}=\sum_{k=1}^{K}\log_{2}(1+\gamma_k)$
and applying the quotient rule to~\eqref{eq:sinr}, the sum-rate gradient is given by
\begin{equation}
\label{eq:sumrate_grad_general}
\nabla_{\boldsymbol{\chi}}R_{\mathrm{sum}}
=
\frac{2}{\ln 2}
\sum_{k=1}^{K}
\frac{
	\Re\!\left\{
	\mathbf{D}_{k,\boldsymbol{\chi}}^{H}
	\mathbf{w}_{k}
	\zeta_k^{\mathrm{sig}}
	-
	\sum_{j\neq k}
	\mathbf{D}_{j,\boldsymbol{\chi}}^{H}
	\mathbf{w}_{k}
	\zeta_{k,j}^{\mathrm{int}}
	\right\}
}{
	1+\gamma_k
}.
\end{equation}
We next evaluate $\mathbf D_{j,\boldsymbol{\chi}}$ for the BS boresight update
$\boldsymbol{\chi}=\mathbf f_m$ and the IRS orientation update
$\boldsymbol{\chi}=\boldsymbol{\psi}$, respectively.

\subsubsection{RA Boresight Update}

For fixed $\boldsymbol{\psi}$, the RA boresight vectors
$\{\mathbf f_m\}_{m=1}^{M}$ are updated sequentially.
For antenna~$m$, the gradient is obtained by setting
$\boldsymbol{\chi}=\mathbf f_m$ in~\eqref{eq:sumrate_grad_general}.
From~\eqref{eq:hkB_total} and~\eqref{eq:HRB_entry},
$\mathbf f_m$ enters the composite channel in~\eqref{eq:composite_channel}
only through $[\mathbf g_{\mathrm B,k,\ell}]_m$ and $G_{m,n}^{\mathrm{BS}}$.
Hence, $\partial\mathbf h_k/\partial\mathbf f_m^{T}
\in\mathbb C^{M\times 3}$ has only the $m$-th row nonzero, and the
nonzero row is given by~\eqref{eq:mu_Dk_fm_row}, at the top of the next page,
\begin{figure*}[!t]
\vspace*{1pt}
\begin{equation}
	\label{eq:mu_Dk_fm_row}
	\begin{aligned}
		\frac{\partial[\mathbf h_k]_m}{\partial\mathbf f_m^{T}}
		={}&
		\frac{\partial[\mathbf h_{\mathrm B,k}]_m}
		{\partial\mathbf f_m^{T}}
		+
		\sum_{n=1}^{N}
		v_n[\mathbf h_{\mathrm R,k}]_n
		\frac{\partial[\mathbf H_{\mathrm{RB}}]_{m,n}}
		{\partial\mathbf f_m^{T}} \\
		={}&
		\sum_{\ell=1}^{L_{\mathrm{B},k}}
		c_{\mathrm{B},k,\ell}\,
		[\mathbf{a}_{\mathrm{B},k,\ell}]_m\,
		\frac{\partial\sqrt{[\mathbf{g}_{\mathrm{B},k,\ell}]_m}}
		{\partial\mathbf f_m^{T}} 
		+
		\sum_{n=1}^{N}
		\frac{\lambda\, v_n[\mathbf{h}_{\mathrm{R},k}]_n
			\sqrt{G^{\mathrm{ref}}_{n,m}(\boldsymbol{\psi})}}
		{4\pi\, d_{n,m}(\boldsymbol{\psi})}\,
		e^{-jk_c d_{n,m}(\boldsymbol{\psi})}\,
		\frac{\partial\sqrt{G^{\mathrm{BS}}_{m,n}}}
		{\partial\mathbf f_m^{T}} .
	\end{aligned}
\end{equation}
\hrulefill
\vspace*{1pt}
\end{figure*}
where $\partial\sqrt{[\mathbf g_{\mathrm B,k,\ell}]_m}/
\partial\mathbf f_m^{T}$ and $\partial\sqrt{G^{\mathrm{BS}}_{m,n}}/
\partial\mathbf f_m^{T}$ are obtained by removing the angular
Jacobian $\partial\mathbf f_m/\partial\boldsymbol{\theta}_m$ in
\eqref{eq:su_RA_gain_derivatives} and \eqref{eq:su_RA_gain_derivatives2}.
Consistent with~\eqref{eq:gain_pattern_modeling}, the above derivatives
are evaluated only when the corresponding directional cosines are positive.
For directions outside the front hemisphere, both the gain and its
derivative are set to zero.
Substituting~\eqref{eq:mu_Dk_fm_row} into
\eqref{eq:sumrate_grad_general} yields $\nabla_{\mathbf f_m}R_{\mathrm{sum}}$.
Since the multi-user update is performed with respect to the boresight
vector $\mathbf f_m$, the elevation constraint in~\eqref{eq:P1_const_elev}
is equivalently written as
\begin{equation}
\label{eq:RA_spherical_cap}
\mathcal F_m
\triangleq
\left\{
\mathbf f\in\mathbb R^{3}
\mid
\|\mathbf f\|=1,\;
\mathbf f^{T}\mathbf f_{\mathrm{ref}}
\geq
\cos\theta_{\max}
\right\}, 
\end{equation}
where $\mathbf f_{\mathrm{ref}}=[0,1,0]^{T}$.
For a trial step size $s_m>0$, the projected-gradient candidate is
\begin{equation}
\label{eq:boresight_vector_update}
\mathbf f_m^{\mathrm{cand}}(s_m)
=
\Pi_{\mathcal F_m}
\!\left(
\mathbf f_m+s_m\nabla_{\mathbf f_m}R_{\mathrm{sum}}
\right),
\end{equation}
where $\Pi_{\mathcal F_m}(\cdot)$ denotes the projection onto
$\mathcal F_m$.
Starting from $s_m=s_{\max}$, backtracking with factor
$\delta_{\mathrm{back}}\in(0,1)$ is performed until
$\mathbf f_m^{\mathrm{cand}}(s_m)$ yields a non-decreasing true sum rate
with the other variables fixed.
If no acceptable step is found before $s_m<s_{\min}$, the current
$\mathbf f_m$ is retained.
After the RA boresight update, the deflection angles
$\{\boldsymbol{\theta}_m\}_{m=1}^{M}$ are recovered from
$\{\mathbf f_m\}_{m=1}^{M}$ via the inverse mapping
of~\eqref{eq:boresight_vector}.

\subsubsection{IRS Orientation Update}

Given the updated RA boresights $\boldsymbol{\Theta}$, the IRS
orientation $\boldsymbol{\psi}$ is updated with fixed
$\mathbf W$ and $\mathbf v$.
The ascent direction is obtained by setting
$\boldsymbol{\chi}=\boldsymbol{\psi}$ in
\eqref{eq:sumrate_grad_general}.
For each Euler angle $\psi_i$, the channel derivative
$\partial\mathbf h_k/\partial\psi_i$ follows from the two-term
decomposition in~\eqref{eq:su_channel_derivative_psi} after restoring
the user index $k$.
The corresponding constituent derivatives are given by
\eqref{eq:su_hR_derivative}, \eqref{eq:su_HRB_derivative},
\eqref{eq:su_Gref_GBS_derivative}, \eqref{eq:su_Gref_GBS_derivative2} and~\eqref{eq:su_d_derivative}.
Stacking the derivatives for $i=1,2,3$ and substituting them into
\eqref{eq:sumrate_grad_general} yields
\begin{equation}
\label{eq:mu_psi_gradient}
\mathbf g^{(t)}
\triangleq
\nabla_{\boldsymbol{\psi}}R_{\mathrm{sum}}
\bigl(\boldsymbol{\psi}^{(t-1)}\bigr)
\in\mathbb R^{3},
\end{equation}
where $\boldsymbol{\psi}^{(t-1)}$  denotes the IRS orientation at iteration $t-1$.
For the multi-user sum-rate update, the visibility constraint is
linearized before line search, and the trial point is projected onto the
intersection of the Euler-angle box and the linearized visibility
constraint by Dykstra's method \cite{bauschke1994dykstra}.
The Euler-angle box and the locally linearized visibility constraint are
respectively given by
\begin{equation}
\label{eq:mu_psi_constraint_sets}
\begin{aligned}
	\!\!\!\!	\mathcal B_{\psi}
	&\!=\!
	\left\{
	\boldsymbol{\psi}
	\mid
	-\boldsymbol{\psi}_{\max}
	\preceq
	\boldsymbol{\psi}
	\preceq
	\boldsymbol{\psi}_{\max}
	\right\},\\
	\!\!\!\!	\mathcal L_{\psi}^{(t)}
	&\!=\!
	\left\{
	\boldsymbol{\psi}
	\mid
	g_{\mathrm{vis}}(\boldsymbol{\psi}^{(t-1)})
	\!+\!
	\bigl(\mathbf{a}_{\mathrm{vis}}^{(t)}\bigr)^{T}
	\!\!\!\!\left(
	\boldsymbol{\psi}-\boldsymbol{\psi}^{(t-1)}
	\right)
	\geq0
	\right\},
\end{aligned}
\end{equation}
where $\mathbf{a}_{\mathrm{vis}}^{(t)}\triangleq\nabla_{\boldsymbol{\psi}}g_{\mathrm{vis}}(\boldsymbol{\psi}^{(t-1)})\in\mathbb{R}^{3}$ denotes the constraint gradient at $\boldsymbol{\psi}^{(t-1)}$.
Let
$\mathcal C_{\psi}^{(t)}
=
\mathcal B_{\psi}\cap\mathcal L_{\psi}^{(t)}$, which is used as the projection set for the
IRS-orientation candidate.
Different from the single-user IRS-orientation update, a safeguarded BB \cite{barzilai1988bb}
stepsize is adopted here to account for the stronger local variation of
the multi-user sum-rate objective caused by inter-user interference.
Denote the resulting stepsize by $\tau^{(t)}$.
The projected candidate is then given by
\begin{equation}
\label{eq:mu_psi_projected_update}
\boldsymbol{\psi}^{\mathrm{cand}}
=
\Pi_{\mathcal C_{\psi}^{(t)}}
\left(
\boldsymbol{\psi}^{(t-1)}
+
\tau^{(t)}\mathbf g^{(t)}
\right),
\end{equation}
where $\Pi_{\mathcal C_{\psi}^{(t)}}(\cdot)$ is implemented by Dykstra's projection.
Since $\mathcal L_{\psi}^{(t)}$ is a local approximation of the visibility constraint, the candidate is further checked using the original constraint $g_{\mathrm{vis}}(\boldsymbol{\psi})\geq0$. The candidate is accepted only if the original visibility constraint is satisfied and the true sum rate is non-decreasing with $\mathbf W$, $\mathbf v$, and $\boldsymbol{\Theta}$ fixed. Otherwise, $\tau^{(t)}$ is reduced and the projection is repeated. If no acceptable candidate is found before $\tau^{(t)}<\tau_{\min}$, $\boldsymbol{\psi}^{(t-1)}$ is retained.
The overall procedure is summarized in Algorithm~\ref{alg:rotation_opt}.

\begin{algorithm}[!tbp]
\caption{Rotation Update for Fixed $\mathbf W$ and $\mathbf v$}
\label{alg:rotation_opt}
\begin{algorithmic}[1]
	\State \textbf{Input:} Fixed $\mathbf W$ and $\mathbf v$; initial rotations
	$\boldsymbol{\Theta}^{(0)}$ and $\boldsymbol{\psi}^{(0)}$; stopping
	parameters $\{\epsilon_{\mathrm{rot}},\epsilon_0,T_{\mathrm{rot}}\}$;
	RA step parameters $\{s_{\max},s_{\min}\}$; IRS BB parameters
	$\{\tau_{\min},\tau_{\mathrm{ini}},\tau_{\max}\}$; backtracking factor
	$\delta_{\mathrm{back}}$.
	\State Initialize $\mathbf f_m\gets\mathbf f_m^{(0)}$,
	$\boldsymbol{\psi}\gets\boldsymbol{\psi}^{(0)}$,
	and $t\gets1$.
	\Repeat
	\State $R^{\mathrm{prev}}\gets
	R_{\mathrm{sum}}(\{\mathbf f_m\}_{m=1}^{M},\boldsymbol{\psi})$.
	\For{$m=1,\ldots,M$}
	\State Compute $\nabla_{\mathbf f_m}R_{\mathrm{sum}}$ by
	\eqref{eq:sumrate_grad_general}.
	\State Update $\mathbf f_m$ by projected backtracking using
	\eqref{eq:boresight_vector_update}.
	\EndFor
	\State Construct $\mathbf g^{(t)}$ and $\mathcal C_{\psi}^{(t)}$ by~\eqref{eq:mu_psi_gradient} and
	\eqref{eq:mu_psi_constraint_sets}, respectively.
	\State Choose $\tau^{(t)}$ by the safeguarded BB rule.
	\State Update $\boldsymbol{\psi}$ by~\eqref{eq:mu_psi_projected_update}.
	\State Recover $\{\boldsymbol{\theta}_m\}_{m=1}^{M}$ from
	$\{\mathbf f_m\}_{m=1}^{M}$ via~\eqref{eq:boresight_vector}.
	\State $R\gets
	R_{\mathrm{sum}}(\{\mathbf f_m\}_{m=1}^{M},\boldsymbol{\psi})$,
	and $t\gets t+1$.
	\Until{$|R-R^{\mathrm{prev}}|/
		\max\{R^{\mathrm{prev}},\epsilon_0\}\leq\epsilon_{\mathrm{rot}}$
		or $t\geq T_{\mathrm{rot}}$}
	\State \textbf{return}
	$\boldsymbol{\Theta}^{\star}\gets
	\{\boldsymbol{\theta}_m\}_{m=1}^{M}$ and
	$\boldsymbol{\psi}^{\star}\gets\boldsymbol{\psi}$.
\end{algorithmic}
\end{algorithm}

\subsection{Overall AO Algorithm and Complexity Analysis}
\label{subsec:mu_overall}
\begin{algorithm}[!tbp]
\caption{Proposed AO Algorithm for Problem~(P1)}
\label{alg:mu_ao}
\begin{algorithmic}[1]
	\State \textbf{Input:} Initial $\mathbf{v}^{(0)}$, $\boldsymbol{\Theta}^{(0)}$, and $\boldsymbol{\psi}^{(0)}$, tolerance $\epsilon_{\mathrm{AO}}$, and maximum iteration number $T_{\mathrm{AO}}$.
	\State \textbf{Initialization:} Set $t\gets0$.
	\State Build $\{\mathbf{h}_k^{(0)}\}_{k=1}^{K}$ from $\boldsymbol{\Theta}^{(0)}$, $\boldsymbol{\psi}^{(0)}$, and $\mathbf{v}^{(0)}$.
	\State Update $\mathbf{W}^{(0)}$ via~\eqref{eq:mmse_combiner}.
	\State Compute $R_{\mathrm{sum}}^{(0)}\gets\sum_{k=1}^{K}\log_2(1+\gamma_k^{(0)})$.
	\Repeat
	\State Update $\mathbf{v}^{(t+1)}$ by solving~\eqref{eq:mu_qcqp} via~\eqref{eq:mu_Q_matrix}--\eqref{eq:FP_Euclidean_grad}.
	\State Update $\boldsymbol{\Theta}^{(t+1)}$ and $\boldsymbol{\psi}^{(t+1)}$ by Algorithm~\ref{alg:rotation_opt}.
	\State Update $\mathbf{W}^{(t+1)}$ via~\eqref{eq:mmse_combiner}.
	\State Compute $R_{\mathrm{sum}}^{(t+1)}\gets \sum_{k=1}^{K}\log_2(1+\gamma_k^{(t+1)})$.
	\State Set $t\gets t+1$;
	\Until{$|R_{\mathrm{sum}}^{(t)}-R_{\mathrm{sum}}^{(t-1)}|/R_{\mathrm{sum}}^{(t-1)}\leq \epsilon_{\mathrm{AO}}$ or $t\geq T_{\mathrm{AO}}$}
	\State \textbf{return} $\mathbf{W}^{(t)}$, $\mathbf{v}^{(t)}$, $\boldsymbol{\Theta}^{(t)}$, and $\boldsymbol{\psi}^{(t)}$.
\end{algorithmic}
\end{algorithm}

The three-block AO procedure is summarized in Algorithm~\ref{alg:mu_ao}, where $\epsilon_{\mathrm{AO}}$ and $T_{\mathrm{AO}}$ denote the stopping tolerance and the maximum number of AO iterations, respectively.
By iteratively optimizing each variable, the accepted sum-rate values are monotonically non-decreasing over iterations.
Since the transmit powers are finite, the objective value of problem~(P1) is upper-bounded.
Therefore, the objective value generated by Algorithm~\ref{alg:mu_ao} converges.

The per-AO-iteration complexity is dominated by the IRS phase-shift update and the rotation update.
The MMSE combiner in~\eqref{eq:mmse_combiner} requires $\mathcal{O}(KM^3+K^2M^2)$ operations, where the cost comes from the covariance-matrix inversions and composite-channel inner products.
The FP-RCG phase-shift update requires $\mathcal{O}(T_{\mathrm{FP}}(K^2MN+K^2N^2+T_{\mathrm{RCG}}N^2))$ operations.
The $K^2MN$ and $K^2N^2$ terms come from the construction of $\mathbf{Q}^{(i)}$ and $\mathbf{q}^{(i)}$ in~\eqref{eq:mu_Q_matrix}-\eqref{eq:mu_q_vector}, while the $T_{\mathrm{RCG}}N^2$ term comes from matrix-vector products with $\mathbf{Q}^{(i)}$ during the RCG iterations.
The rotation update requires $\mathcal{O}(T_{\mathrm{rot}}KMN)$ operations, where the cost comes from the gradient evaluation of the BS boresights and IRS orientation.
Therefore, the complexity over $T_{\mathrm{AO}}$ iterations is $\mathcal{O}(T_{\mathrm{AO}}(KM^3+K^2M^2+T_{\mathrm{FP}}(K^2MN+K^2N^2+T_{\mathrm{RCG}}N^2)+T_{\mathrm{rot}}KMN))$.

\section{Simulation Results}
\label{sec:results}

In this section, numerical simulations are conducted to evaluate
the performance of the proposed algorithm. Let $k_c=2\pi/\lambda$ denote the wavenumber.
Unless otherwise specified, we set $f_c=6$~GHz, $\lambda=0.05$~m, $\mathbf{b}_0=(0,0,10)^T$~m, $\mathbf{r}_0=(2,4,14)^T$~m, $K=4$, $P_k=20$~dBm, $\sigma^2=-80$~dBm, $L_{\mathrm{B},k}=L_{\mathrm{IR},k}=4$, $\kappa_{\mathrm{NLoS}}=0.3$, $p_{\mathrm{BS}}=6$, $p_{\mathrm{IRS}}=5$, and $\theta_{\max}=\alpha_{\max}=\beta_{\max}=\phi_{\max}=60^\circ$.
The BS array is a $4\times4$ UPA with $M=16$ and $d_{\mathrm{B}}=\lambda/2$. 
The IRS adopts a fixed square aperture with side length $0.5$~m, element spacing $d_{\mathrm{IRS}}=\lambda/2$, and $N_y=N_z=21$, resulting in $N=441$ reflecting elements.
User azimuth angles are drawn uniformly from $[-60^\circ,60^\circ]$.
The user-distance range is chosen such that the user-IRS channel remains in the far field.
The tested IRS-BS distances are chosen in the spherical-wave near-field regime considered by the IRS-BS channel model.
The direct user-BS channel is retained in all simulations and is attenuated by $25$~dB in the default setting to emulate partial blockage of the direct link by buildings. We compare the following schemes:
\begin{itemize}
\item \textit{Dual Rotation}: $\{\boldsymbol{\Theta}, \boldsymbol{\psi}, \mathbf{v}, \mathbf{W}\}$ are jointly optimized.
\item \textit{BS-only}: $\{\boldsymbol{\Theta}, \mathbf{v}, \mathbf{W}\}$ are optimized, while the IRS orientation is fixed at the BS-facing reference orientation.
\item \textit{IRS-only}: $\{\boldsymbol{\psi}, \mathbf{v}, \mathbf{W}\}$ are optimized, while all RA boresights are fixed at the reference direction $\mathbf{f}_{\mathrm{ref}}=[0,1,0]^T$.
\item \textit{Fixed}: $\{\mathbf{v}, \mathbf{W}\}$ are optimized, while the IRS orientation is fixed at the BS-facing reference orientation and all RA boresights are fixed at $\mathbf{f}_{\mathrm{ref}}$.
\end{itemize}

\subsection{Convergence Behavior and Performance Scaling}
\label{subsec:convergence}

\begin{figure}[t]
\centering
\includegraphics[width=0.87 \linewidth]{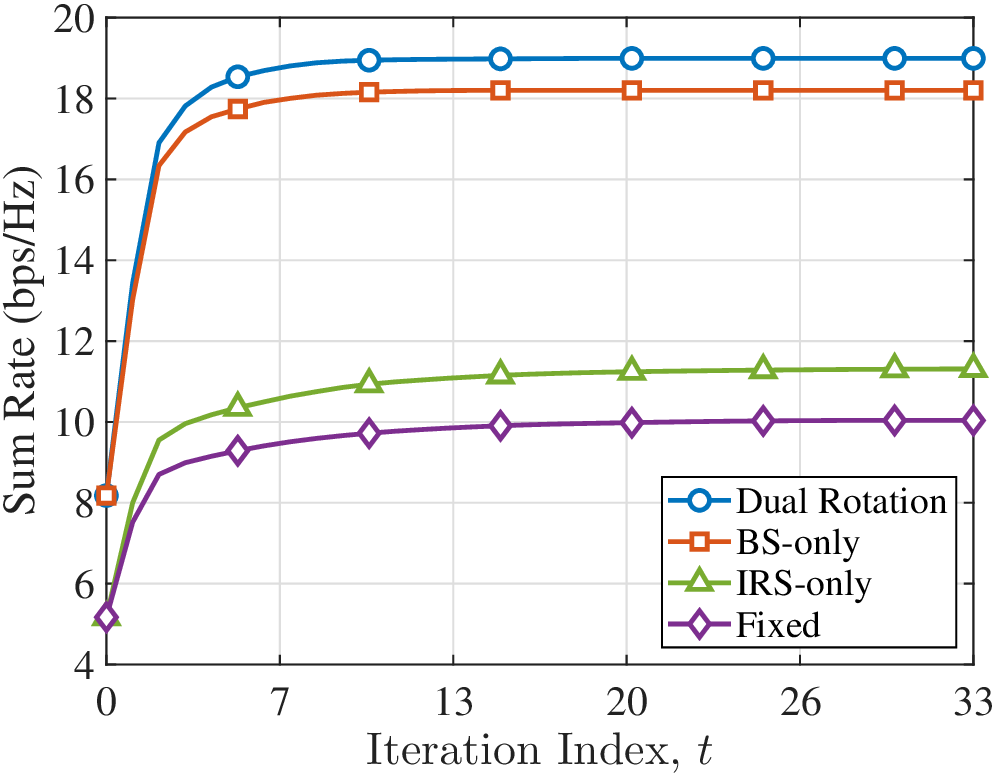}
\caption{Convergence behavior of the proposed AO algorithm.}
\label{fig:convergence}
\end{figure}
Fig.~\ref{fig:convergence} shows the convergence behavior of the AO-based algorithms under different schemes.
It is observed that the sum-rate of all schemes are non-decreasing and converge within about $20$ iterations, which verifies the monotonic convergence behavior of the proposed update procedure.
This is because each block of variables is updated alternately and the new solution is accepted only when the objective value is non-decreasing.
The converged sum rates of \textit{Dual Rotation}, \textit{BS-only}, \textit{IRS-only}, and \textit{Fixed} are about $18.99$, $18.20$, $11.31$, and $10.04$~bps/Hz, respectively.
The proposed \textit{Dual Rotation} scheme achieves the highest converged sum rate, which demonstrates the benefit of jointly optimizing the BS antenna boresights and IRS panel orientation.

\begin{figure}[t]
\centering
\includegraphics[width=0.87 \linewidth]{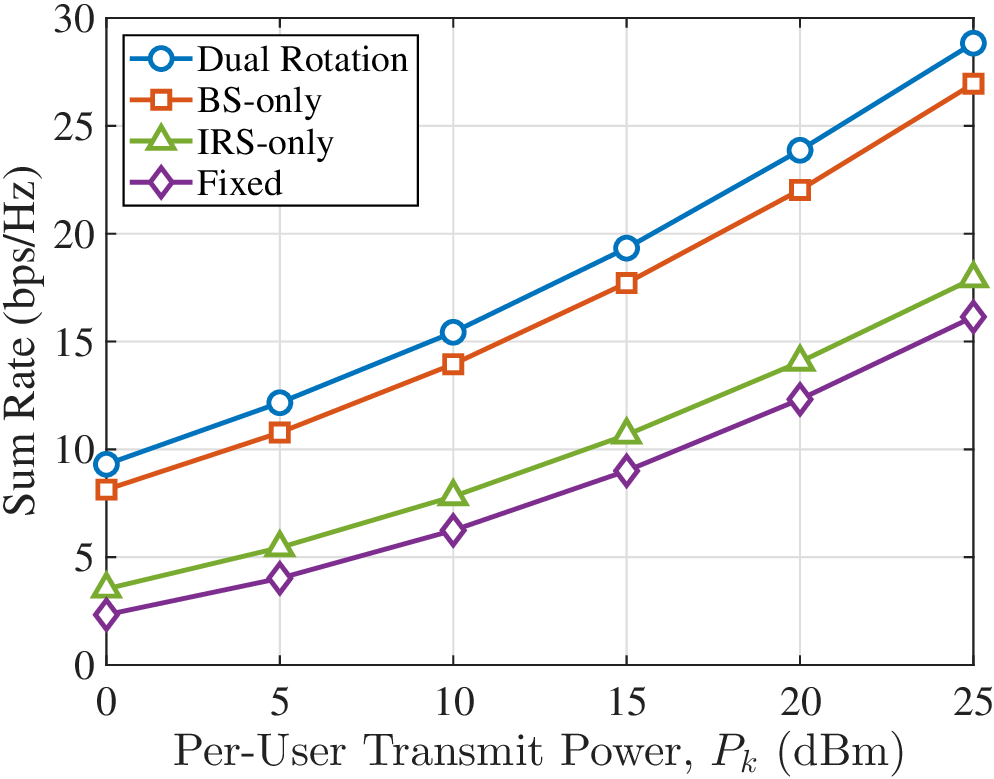}
\caption{Sum rate versus per-user transmit power.}
\label{fig:power}
\end{figure}
Fig.~\ref{fig:power} shows the sum rate versus the per-user transmit power $P_k$ under different schemes.
It can be observed that the sum rate of all schemes increases with $P_k$, and the proposed \textit{Dual Rotation} scheme achieves the highest sum rate over the whole transmit-power range.
For example, when $P_k=0$~dBm, the sum rates achieved by \textit{Dual Rotation}, \textit{BS-only}, \textit{IRS-only}, and \textit{Fixed} are about $9.31$, $8.13$, $3.53$, and $2.33$~bps/Hz, respectively.
When $P_k=25$~dBm, the corresponding sum rates increase to about $28.84$, $26.95$, $17.92$, and $16.14$~bps/Hz.
This is because increasing the transmit power improves both the direct and IRS-reflected signal components.
Moreover, the \textit{BS-only} and \textit{IRS-only} schemes outperform the \textit{Fixed} benchmark, which demonstrates the benefits of BS boresight steering and IRS orientation control, respectively.
The proposed \textit{Dual Rotation} scheme further outperforms the two single-rotation schemes, indicating the benefit of jointly optimizing BS antenna boresights and IRS panel orientation.

\begin{figure}[t]
\centering
\includegraphics[width= 0.87 \linewidth]{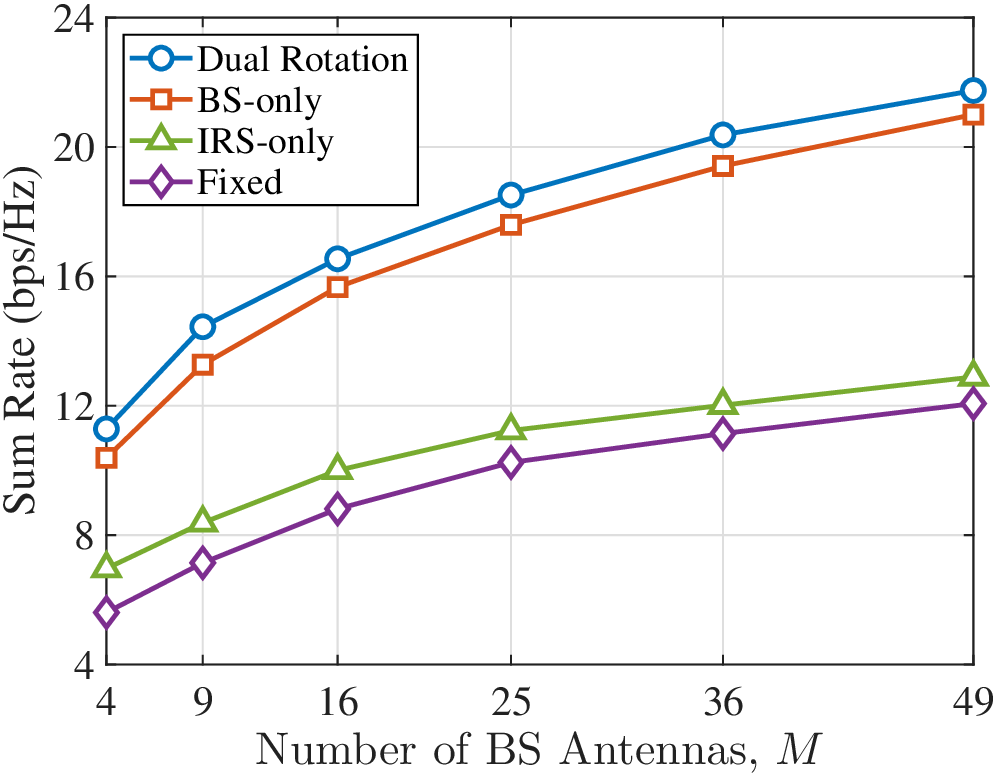}
\caption{Sum rate versus the number of BS antennas $M$.}
\label{fig:antennas}
\end{figure}
Fig.~\ref{fig:antennas} shows the sum rate versus the number of BS antennas, where $M\in\{4,9,16,25,36,49\}$, $K=4$, $P_k=20$~dBm, and $N=441$.
It can be observed that the sum rates of all schemes increase with $M$.
For example, when $M=4$, the sum rates achieved by \textit{Dual Rotation}, \textit{BS-only}, \textit{IRS-only}, and \textit{Fixed} are about $11.28$, $10.39$, $6.97$, and $5.61$~bps/Hz, respectively.
When $M$ increases to $49$, the corresponding sum rates are about $21.75$, $21.00$, $12.89$, and $12.07$~bps/Hz.
This is because a larger BS array provides more receive dimensions for separating multi-user signals and collecting the IRS-reflected components.
Moreover, the performance gaps between \textit{Dual Rotation} and the \textit{IRS-only}/\textit{Fixed} benchmarks become larger in absolute value as $M$ increases, while \textit{Dual Rotation} still provides additional gain over the strong \textit{BS-only} benchmark.
This indicates that the proposed joint rotation design can better exploit the enlarged BS array.

\begin{figure}[t]
\centering
\includegraphics[width=  0.87 \linewidth]{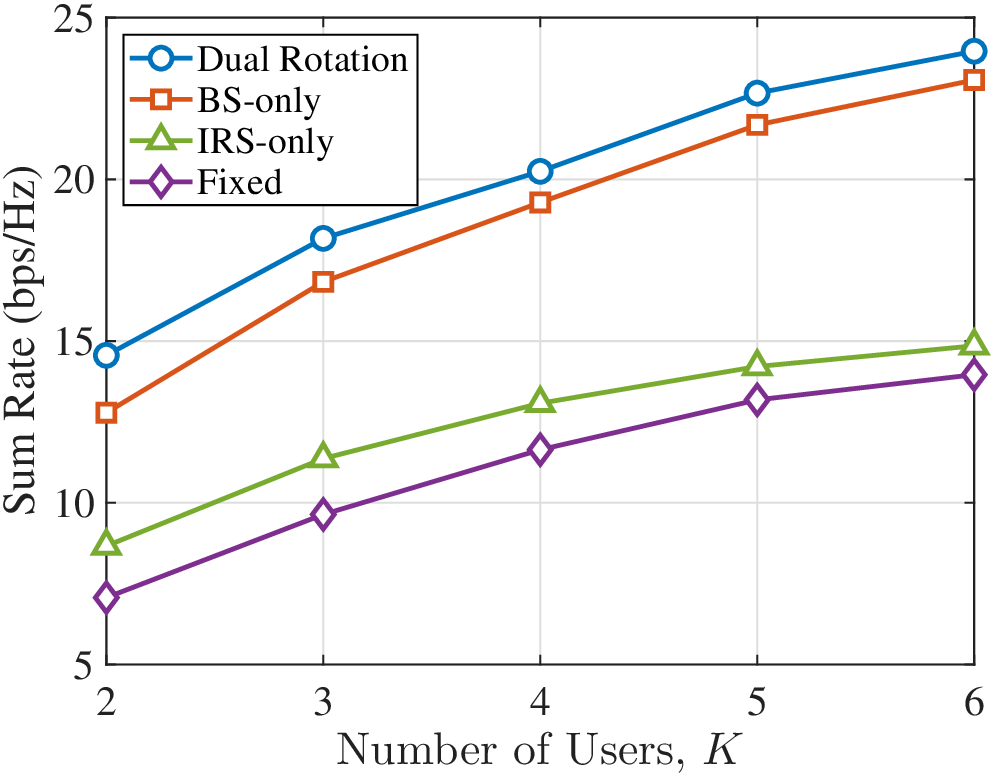}
\caption{Sum rate versus the number of users. }
\label{fig:users}
\end{figure}
Fig.~\ref{fig:users} shows the sum rate versus the number of users $K$, where $K\in\{2,\ldots,6\}$, $M=16$, $P_k=20$~dBm, and $N=441$.
It can be observed that the sum rates of all schemes increase with $K$.
For example, when $K=2$, the sum rates achieved by \textit{Dual Rotation}, \textit{BS-only}, \textit{IRS-only}, and \textit{Fixed} are about $14.56$, $12.78$, $8.67$, and $7.07$~bps/Hz, respectively.
When $K$ increases to $6$, the corresponding sum rates become about $23.96$, $23.07$, $14.85$, and $13.96$~bps/Hz.
This is because each user transmits with fixed power, and increasing $K$ introduces more data streams and a higher total uplink transmit power.
Meanwhile, the sum-rate growth tends to slow down when $K$ becomes large, since more users also introduce stronger inter-user interference under the fixed BS array size.
The proposed \textit{Dual Rotation} scheme consistently outperforms the benchmark schemes for all considered values of $K$, which shows its robustness under different user loading levels.

\subsection{Near-Field Rotation Coupling Validation}

\begin{figure}[t]
\centering
\includegraphics[width=0.9  \linewidth]{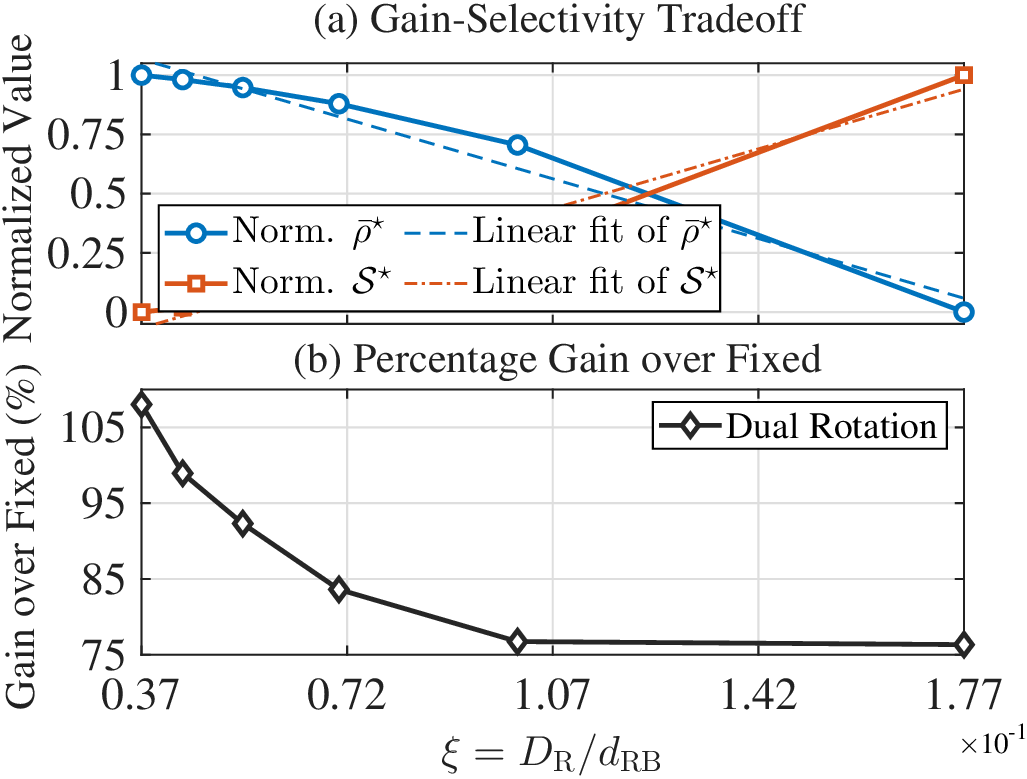}
\caption{Near-field rotation-coupling metrics versus $\xi$.}
\label{fig:nearfield_tradeoff}
\end{figure}

Fig.~\ref{fig:nearfield_tradeoff} shows the near-field rotation-coupling metrics versus the normalized IRS aperture-to-distance ratio $\xi$.
In Fig.~\ref{fig:nearfield_tradeoff}(a), $\bar{\rho}^{\star}$ and $\mathcal{S}^{\star}$ are normalized to $[0,1]$ for ease of comparison.
It is observed that the normalized $\bar{\rho}^{\star}$ decreases as $\xi$ increases, whereas the normalized $\mathcal{S}^{\star}$ increases as $\xi$ increases.
This indicates that a larger IRS aperture-to-distance ratio leads to stronger spatial variation in the IRS-BS propagation directions.
As a result, maintaining a high average alignment over the IRS-BS link becomes more difficult, while the reflection-gain variation becomes more pronounced.
This observation is consistent with Proposition~\ref{prop:xi_tradeoff}.
Fig.~\ref{fig:nearfield_tradeoff}(b) reports the percentage gain of \textit{Dual Rotation} with respect to the \textit{Fixed} baseline.
It is observed that the gain is positive over the considered range of $\xi$.
The gain curve generally decreases as $\xi$ increases, indicating that the relative gain of \textit{Dual Rotation} over the fixed-orientation baseline becomes less pronounced for larger IRS aperture-to-distance ratios.
This result shows that coordinated rotation remains effective under the considered near-field IRS-BS geometries.

\begin{figure}[t]
\centering
\includegraphics[width=0.95    \linewidth]{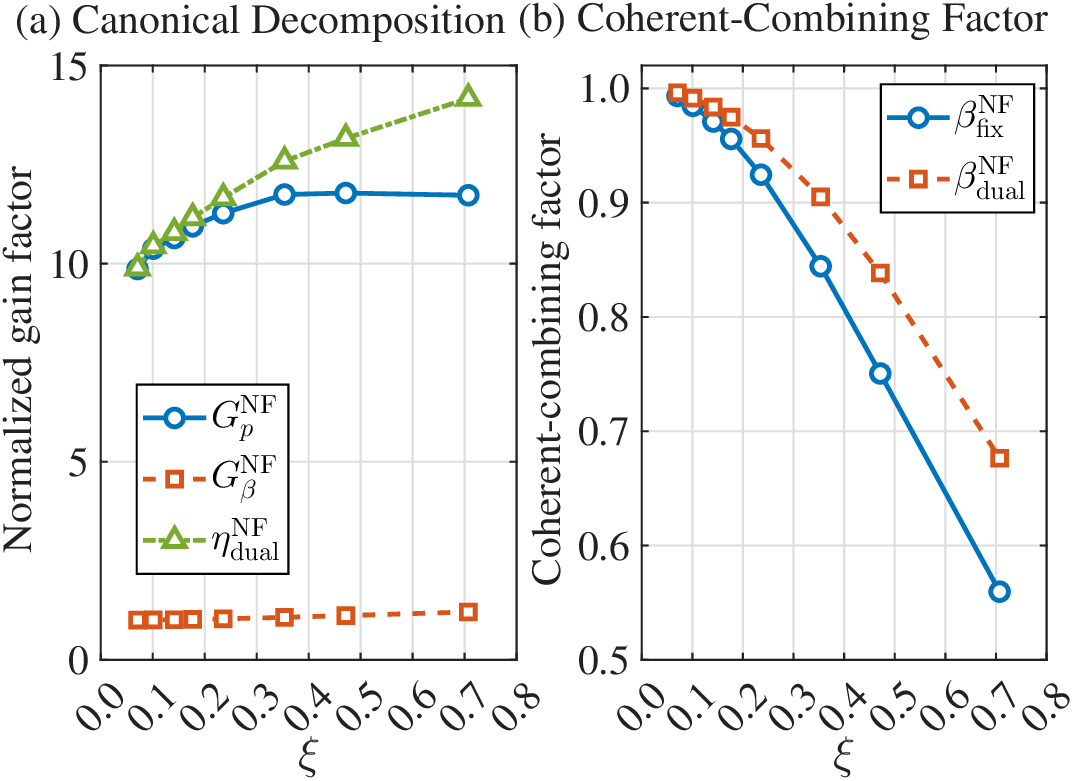}
\caption{Reflected-channel decomposition versus $\xi$.}
\label{fig:single_user_nf_decomposition}
\end{figure}
Fig.~\ref{fig:single_user_nf_decomposition} illustrates the reflected-channel decomposition in Proposition~\ref{prop:nf_dual_coupling} under the cascaded-LoS setting without the direct link or NLoS components. 
Fig.~\ref{fig:single_user_nf_decomposition}(a) plots $G_{p}^{\mathrm{NF}}$, $G_{\beta}^{\mathrm{NF}}$, and $\eta_{\mathrm{dual}}^{\mathrm{NF}}$ versus $\xi$. 
It is observed that $G_{p}^{\mathrm{NF}}$ is much larger than $G_{\beta}^{\mathrm{NF}}$ over the considered range, indicating that the near-field dual-rotation gain mainly comes from the increase in the average reflected-channel power. 
For example, when $\xi$ increases from $0.071$ to $0.707$, $G_{p}^{\mathrm{NF}}$ stays around $9.86$-$11.78$, whereas $G_{\beta}^{\mathrm{NF}}$ increases from about $1.003$ to $1.209$. 
As a result, $\eta_{\mathrm{dual}}^{\mathrm{NF}}$ increases from about $9.89$ to $14.17$. 
This shows that the coherent-combining gain is relatively small compared with the column-power gain, but becomes more visible when the near-field effect becomes stronger.
Fig.~\ref{fig:single_user_nf_decomposition}(b) further compares the coherent-combining factors $\beta_{\mathrm{fix}}^{\mathrm{NF}}$ and $\beta_{\mathrm{dual}}^{\mathrm{NF}}$. 
Both factors decrease with $\xi$, which indicates that stronger near-field propagation makes coherent combining across IRS-reflected components more difficult. 
However, $\beta_{\mathrm{dual}}^{\mathrm{NF}}$ is consistently larger than $\beta_{\mathrm{fix}}^{\mathrm{NF}}$. 
For example, at $\xi=0.707$, $\beta_{\mathrm{dual}}^{\mathrm{NF}}$ is about $0.675$, while $\beta_{\mathrm{fix}}^{\mathrm{NF}}$ is about $0.560$. 
This shows that coordinated dual rotation can partly compensate for the coherent-combining loss caused by the near-field IRS-BS channel.

\section{Conclusion}
\label{sec:conclusion}

This paper investigated coordinated BS boresight steering and IRS panel orientation for an IRS-assisted multi-user MIMO uplink system.
We formulated a sum-rate maximization problem that jointly optimizes the receive beamforming, IRS phase shifts, BS antenna boresights, and IRS panel orientation.
A single-user case was analyzed to characterize the dual-rotation gain under far-field and near-field IRS-BS propagation conditions.
For the general multi-user case, we developed an AO algorithm, where the receive beamforming was obtained in closed form, the IRS phase shifts were optimized via FP-assisted RCG, and the rotation variables were updated by projected gradient methods.
Numerical results confirmed that coordinated dual rotation consistently improves the uplink sum rate over fixed-orientation and single-rotation benchmarks.
The impact of the aperture-to-distance ratio on the alignment variation and reflection-gain distribution was further illustrated across different IRS-BS distances.
These findings suggest that IRS panel orientation is an important optimization variable for coordinated dual-rotation design.

\bibliographystyle{IEEEtran}
\bibliography{Multi_User_MIMO_with_Rotatable_Antennas_and_IR_Joint_Antenna_Boresight_and_IRS_Orientation_Design}

\end{document}